\documentclass[11pt]{article}
\usepackage{amsmath,amssymb,amsbsy,amsfonts,amsthm,latexsym,
               amsopn,amstext,amsxtra,euscript,amscd,color}
\usepackage{hyperref} 
\usepackage{marginnote}
\usepackage{float}
\usepackage{makecell}

\usepackage{graphicx}
\usepackage{longtable}
\usepackage[textsize=scriptsize]{todonotes}

\newtheorem{theorem}{Theorem}[section]
\newtheorem{lemma}[theorem]{Lemma}
\newtheorem{corollary}[theorem]{Corollary}

\newtheorem{proposition}[theorem]{Proposition}
\newtheorem{remark}[theorem]{Remark}
\newtheorem{definition}[theorem]{Definition}
\newtheorem{example}[theorem]{Example}

\newtheorem{problem}[theorem]{Problem}

\def\F{{\mathbb F}}

\def\F{{\mathbb F}}

\def\00{{\bf 0}}
\def\11{{\bf 1}}
\def\+{\oplus}

\def\\{\cr}
\def\({\left(}
\def\){\right)}

\newcommand{\PcN}{\mathrm{P}c\mathrm{N}}

\providecommand{\newoperator}[3]{%
  \newcommand*{#1}{\mathop{#2}#3}}

\newoperator{\FD}{\mathrm{FD}}{\nolimits}

\begin{document}
\title{\bf Permutation Polynomials Under Multiplicative-Additive Perturbations: Characterization via Difference Distribution Tables}
\author{\bf\Large Ranit Dutta$^{\dagger}$, Pantelimon St\u anic\u a$^*$, Bimal Mandal$^{\dagger}$\\
\\
$^{\dagger}$Department of Mathematics, Indian Institute of Technology Jodhpur\\ 
Karwar--342030, India; duttaranit628@gmail.com, bimalmandal@iitj.ac.in\\
\\
$^*$Naval Postgraduate School, Applied Mathematics Department\\
Monterey, CA 93943, USA;  pstanica@nps.edu
}

\date{}
\maketitle

\begin{abstract}
We investigate permutation polynomials $F$ over finite fields $\mathbb{F}_{p^n}$ whose generalized derivative maps $x \mapsto F(x+a) - cF(x)$ are themselves permutations for all nonzero shifts $a$. This property, termed perfect $c$-nonlinearity (P$c$N), represents optimal resistance to $c$-differential attacks—a concern highlighted by recent cryptanalysis of the Kuznyechik cipher variant. We provide the first characterization using the classical difference distribution table (DDT): $F$ is P$c$N if and only if $\Delta_F(a,b) \cdot \Delta_F(a,c^{-1}b) = 0$ for all nonzero $a,b$. This enables verification in $O(p^{2n})$ time given a precomputed DDT, a significant improvement over the naive $O(p^{3n})$ approach. We prove a strict dichotomy for monomial permutations: the derivative $F(x+\alpha) - cF(x)$ is either a permutation for all nonzero shifts or for none, with the general case remaining open. For quadratic permutations, we provide explicit algebraic characterizations. We identify the first class of affine transformations preserving $c$-differential uniformity and derive tight nonlinearity bounds revealing fundamental incompatibility between P$c$N and APN properties. These results position perfect $c$-nonlinearity as a structurally distinct regime within permutation polynomial theory.
\end{abstract}


\section{Introduction}

Permutation polynomials over finite fields are a classical and active topic, with connections to algebraic combinatorics, coding theory, and cryptography. In many situations one is led to study not only whether a function is a permutation, but also whether certain induced \emph{difference maps} exhibit strong bijectivity or controlled multiplicity. This paper focuses on a particularly rigid instance: functions for which all nontrivial generalized derivative maps $x \mapsto F(x+a) - cF(x)$ remain permutations.

The study of such derivative maps has both classical and applied origins. Differential cryptanalysis, introduced by Biham and Shamir \cite{BS91,BS92}, analyzes how differences propagate through ciphers, with Nyberg \cite{Nyberg94} emphasizing low differential uniformity as a security criterion. A multiplicative variant appeared in work of Borisov et al.\ \cite{BJW02}, who studied differentials of the form $(F(cx),F(x))$ to attack certain IDEA variants. Ellingsen et al.\ \cite{EST20} formalized this as the $c$-derivative for $p$-ary functions, introducing $c$-differential uniformity and the $c$-difference distribution table ($c$DDT). The practical significance of this generalization was recently demonstrated by Stănică et al.\ \cite{SDM25}, who mounted successful attacks on the Kuznyechik block cipher (GOST R 34.12-2015) using $c$-differential cryptanalysis, confirming these properties as genuine security concerns for deployed standards.

Over $\mathbb{F}_{p^n}$, a function $F$ has optimal $c$-differential uniformity (called \emph{perfect $c$-nonlinear}, or P$c$N, following \cite{EST20,BT19}) precisely when each outer $c$-derivative
\[
x \longmapsto F(x+a) - cF(x)
\]
is a permutation of $\mathbb{F}_{p^n}$ for every nonzero $a$ and every relevant parameter $c$. This formulation places the subject naturally within permutation polynomial theory and motivates fundamental structural questions: How can one recognize this property efficiently? How does it behave under affine transformations? Which polynomial families satisfy it, and which cannot?

Our contributions address these questions through a new characterization linking P$c$N directly to the classical difference distribution table (DDT), bypassing the need to construct full $c$DDTs. We establish that the property is preserved under suitable affine transformations, prove a strict dichotomy for monomial permutations (either all shifts preserve bijectivity or none do), and provide explicit characterizations for quadratic permutations. These results position perfect $c$-nonlinearity as a descriptor of a highly constrained regime within permutation polynomial theory, with connections to—but largely disjoint from—the classical theory of APN functions.

\subsection{Our contributions}

\textbf{1. Efficient DDT-based characterization (Section \ref{sec-pcn}).} 
We prove that a permutation polynomial $F$ over $\mathbb{F}_{p^n}$ is P$c$N if and only if its classical DDT satisfies $\Delta_F(a,b) \cdot \Delta_F(a, c^{-1}b) = 0$ for all nonzero $a,b$ (Theorem \ref{thm-pn}). This implies that given a precomputed DDT, P$c$N can be verified in $O(p^{2n})$ time instead of $O(p^{3n})$ (Proposition \ref{prop:complexity-improvement}), 
all values of $c$ for which a given function is P$c$N can be enumerated efficiently in $O(p^{3n})$ time, and further
the characterization reveals fundamental incompatibility: an APN permutation can be P$c$N for at most two values of $c$ (namely $c$ and $c^{-1}$), and typically for none (Theorem \ref{thm:pcn-apn-relation}).
 
\noindent
\textbf{2. Dichotomy for monomial permutations (Section \ref{sec-pcn}).}
For monomial functions $F(x) = ax^d$ over $\mathbb{F}_{2^n}$, we establish an all-or-nothing result: either $F(x+\alpha) + cF(x)$ is a permutation for \emph{all} nonzero shifts $\alpha \in \mathbb{F}_{2^n}^*$, or it is a permutation for \emph{none} of them (Theorem \ref{thm:monomial-dichotomy}). This dichotomy follows from homogeneity and has implications for understanding the Kuznyechik attack \cite{SDM25}. However, this property fails for general polynomials—we provide an explicit counterexample—and characterizing the algebraic conditions ensuring this behavior remains open (Problem \ref{prob:dichotomy-conditions}).

\noindent
\textbf{3. Analysis of specific polynomial families (Section \ref{sec-pcn}).}
Using our DDT characterization, we analyze P$c$N properties for several classes: for $F(x) = x^{(2^k+1)/2}$ over $\mathbb{F}_{2^n}$, we derive explicit conditions on $c$ involving trace functions (Theorem~\ref{thm-gh}), 
for quadratic Dembowski--Ostrom permutations, we provide a complete characterization via orthogonal complements of certain subspaces (Theorem~\ref{thm:quadratic-complete}), as well as we establish that the set of ``bad shifts'' (those for which $F(x+a)-cF(x)$ fails to be a permutation) forms a subspace, leading to a geometric interpretation of partial P$c$N behavior (Proposition \ref{prop-subspace}, Theorem~\ref{thm:subspace-dimension}).
 
\noindent
\textbf{4. Affine equivalence and invariance (Section \ref{sec-pcn-affine}).}
We identify the first explicit class of affine transformations preserving $c$-differential uniformity: certain Frobenius automorphisms $L(x) = x^{p^i}$ combined with translations (Proposition \ref{prop:frobenius-invariance}). More generally, we characterize when composing a P$c$N function with a linear permutation $L$ yields another P$c$N function (Theorems \ref{thm-pn-trans} and \ref{const-pcn}), partially resolving questions from Hasan et al.\ \cite{HRS20} concerning invariance under affine transformations.

\noindent
\textbf{5. Nonlinearity constraints (Section \ref{sec-ddt-prop}).}
We prove that P$c$N functions cannot achieve the Sidelnikov--Chabaud--Vaudenay bound when $c$ is primitive and $n$ is odd. We show that the Walsh--Hadamard spectrum exhibits particular sparsity (Corollary \ref{cor-walsh-sparsity}), leading to tight bounds on maximum Walsh coefficients (Theorem \ref{thm:nonlinearity-tight-bound}). This reveals an inherent trade-off: optimizing for $c$-differential resistance constrains nonlinearity.

\subsection{Context and related work}

The study of $c$-differential properties has developed rapidly since 2020, with Ellingsen et al.~\cite{EST20} introducing foundational definitions, Bartoli and Timpanella~\cite{BT19} independently introducing $c$-planar functions (equivalent to P$c$N), and numerous subsequent works \cite{BC20,HRS20,RS20,WLZ20,ZH20,XD21,SRT20,SRT22,RP22,Anbar23,EM24,JS25,BP25,SDM25} analyzing constructions, spectra, and applications. Our work differs by providing the first computational characterization via the classical DDT, enabling practical verification, and by focusing on structural theorems—dichotomies, incompatibilities, and invariances—rather than catalogs of examples.

\subsection{Organization}

Section \ref{sec:notation} presents preliminaries. Section \ref{sec-pcn} develops our DDT characterization, establishes the dichotomy theorem, and analyzes specific function classes. Section \ref{sec-pcn-affine} identifies transformations preserving $c$-differential uniformity and provides construction methods. Section \ref{sec-ddt-prop} derives nonlinearity bounds. Section \ref{concl} concludes with open problems.

\section{Preliminaries}
\label{sec:notation}
Let $n$ be a positive integer and $p$ be a prime integer.
Let $\mathbb F_p$ and $\mathbb F_{p^n}$ be a prime field of characteristic $p$ and an extension field of degree of extension $n$ over $\mathbb F_p$, respectively.
The set of all non-zero elements of $\mathbb F_{p^n}$ is denoted by $\mathbb F_{p^n}^*$.
The cardinality of a set $A$ is denoted by $\#A$.
Let $\mathbb{F}_{p^m}$ be another field extension of $\mathbb{F}_p$ and $m$ divides $n$.
The trace function $\mathrm{Tr}_m^n:\mathbb F_{p^n}\longrightarrow \mathbb F_{p^m}$ is defined as $\mathrm{Tr}_m^n(x)=\sum_{i=0}^{k-1} x^{p^{mi}}$ for all $x\in\mathbb F_{p^n}$ where $k=\frac{n}{m}$.
Any function from $\mathbb F_{p^n}$ to $\mathbb F_{p^m}$ is called a $p$-ary $(n,m)$-function, and the set of $p$-ary $(n,m)$-functions is denoted by $\mathcal B_{n,m}^p$.
If $m=1$, then it is simply called a $p$-ary function, and the set of $p$-ary functions over $\mathbb F_{p^n}$ is denoted by $\mathcal B_n^p$.
If $p=2$, then $2$-ary $(n,m)$-function is simply called an $(n,m)$-function or a S-box or a vectorial Boolean function, and the set of these functions is denoted $\mathcal B_{n,m}$.
If $p=2$ and $m=1$, then it is called a Boolean function in $n$ variables.
The set of $n$-variable Boolean functions is denoted by $\mathcal B_n^2=\mathcal B_n$.
For more details about (vectorial) Boolean functions, we refer to \cite{Carlet10,PS09}.
The non-zero component function of $F\in\mathcal B_{n,m}^p$ is defined by $\mathrm{Tr}_1^n(vF)$, $v\in\mathbb F_{p^m}^*$.
A $p$-ary $(n,n)$-function $F$ can be uniquely represented as a univariate polynomial over $\mathbb F_{p^n}$ of the form $F(x)=\sum_{i=0}^{p^n-1} a_ix^i$, $a_i\in\mathbb F_{p^n}$.
For $p=2$, the algebraic degree of $F$ is the largest hamming weight of the exponents $i$ with $a_i\neq 0$.
The Walsh--Hadamard transform of $f\in\mathcal{B}_n^p$ at $a\in\mathbb F_{p^n}$ is defined by
\begin{equation*}
\mathcal W_f(a)=\sum_{x\in\mathbb F_{p^n}}\zeta^{f(x)-\mathrm{Tr}_1^n(a x)}.
\end{equation*}
where $\zeta=e^{\frac{2\pi\imath}{p}}$, $\imath=\sqrt{-1}$.
The multiset constituted by the values of the Walsh--Hadamard transform, $[\mathcal W_f(a): a\in\mathbb F_{p^n}]$, is called the Walsh--Hadamard spectrum of $f$.
The Walsh--Hadamard transform of a $p$-ary $(n,n)$-function $F$ is defined on  its component functions as  
$\mathcal W_F(a,b)=\sum_{x\in\mathbb F_{p^n}}\zeta^{\mathrm{Tr}_1^n(bF(x)-a x)}.$

The derivative of $F\in\mathcal B_{n,m}^p$ at $a\in\mathbb F_{p^n}$  (denoted by $D_aF$) is defined by $D_aF(x)=F(x+a)-F(x)$, for all $x\in\mathbb F_{p^n}$.
The autocorrelation spectrum of $F$ at $(a,b)\in\mathbb F_{p^n}\times\mathbb F_{p^m}$ is defined as 
\begin{equation*}
\mathcal C_F(a,b)=\sum_{x\in\mathbb F_{p^n}}\zeta^{\mathrm{Tr}_1^m(b(F(x+a)-F(x)))},
\end{equation*}
where $\zeta=e^{\frac{2\pi\imath}{p}}$, $\imath=\sqrt{-1}$.
If $p=2$, then $\zeta=-1$ and $\mathcal C_F(a,b)=2\#\{x\in\mathbb F_{2^n}: \mathrm{Tr}_1^m(b(F(x+a)+F(x)))=0\} -2^n$.
Let us define $S_F(a,b)=\{x\in\mathbb F_{p^n}: F(x+a)-F(x)=b\}$, 
and $\Delta_F(a,b)=\#S_F(a,b)$, where $a\in\mathbb F_{p^n}$ and $b\in\mathbb F_{p^m}$.
We know that $\Delta_F(0,0)=p^n$, and $\Delta_F(0,b)=0$ if $b\neq 0$. If $F$ is permutation over $\mathbb F_{p^n}$, $\Delta_F(a,0)=0$, if $a\neq0$.
The differential uniformity of $F\in\mathcal B_{n,m}^p$ is denoted by $\delta(F)$, defined by $\delta(F)=\max\{\Delta_F(a,b): a\neq0\in\mathbb F_{p^n}, b\in \mathbb F_{p^m}\}$.
If $\delta(F)=1$, then $F$ is called perfect nonlinear (PN), and almost perfect nonlinear when $\delta(F)=2$.
For any $F\in\mathcal B_{n,m}^p$ with $m\leq n$, $\delta(F)\geq p^{n-m}$, and if $p=2$, $\delta(F)\geq 2$ as the solutions are of the form $x, x+a$.

Ellingsen et al. \cite{EST20} generalized the concept of derivative of $p$-ary $(n,m)$-functions, so-called $c$-derivative.
The $c$-derivative, $c\in\mathbb F_{p^m}$, of $F\in\mathcal B_{n,m}^p$ at $a\in\mathbb F_{p^n}$ is defined as  $D_{c,a}F(x)=F(x+a)-cF(x)$, for all  $x\in\mathbb F_{p^n}$ (we proposed in~\cite{SDM25} for this to be called outer $c$-differential and observed that the outer $c$-differential uniformity of a permutation $F$ equals the inner $c$-differential uniformity of $F^{-1}$, as connected to the inner $c$-differential $F(cx+a)-F(x)$).
\begin{theorem}[\cite{SDM25}]
\label{thm:inverse-duality}
Let $F$ be a permutation polynomial over $\mathbb{F}_{p^n}$. For any $a, b, c \in \mathbb{F}_{p^n}$ with $c \neq 0$, the outer $c$-differential entries of $F$ correspond to the inner $c$-differential entries of its inverse $F^{-1}$. Specifically,
\begin{equation*}
\Delta_{c,F}(a,b) = \#\{y \in \mathbb{F}_{p^n} : F^{-1}(cy+b) - F^{-1}(y) = a\}.
\end{equation*}
Consequently, $F$ is P$c$N (with respect to outer $c$-differential) if and only if $F^{-1}$ is P$c$N with respect to the inner $c$-differential.
\end{theorem}

Though not always mentioned, we will work with the outer $c$-differential in this paper.
If $c=1$, it is a usual derivative of $F$ at $a$. If $c=0$, then $D_{0,a}F(x)=F(x+a)$, and if $a=0$, then $D_{c,0}F(x)=(1-c)F(x)$.
Let us define  $S_{c,F}(a,b)=\{x\in\mathbb F_{p^n}: F(x+a)-cF(x)=b\}$ and $\Delta_{c,F}(a,b)=\# S_{c,F}(a,b)$, where $a\in\mathbb F_{p^n}$ and $b,c\in\mathbb F_{p^m}$. Ellingsen et al.
\cite{EST20} introduced an extended difference distribution table (DDT) of vectorial $p$-ary functions, called $c$-difference distribution table ($c$DDT).
If $c=1$, then the $c$DDT coincides with the usual DDT. For each distinct value
$c \in \mathbb{F}_{p^m} \setminus \{0,1\}$, we obtain a different difference table,
whose entry at $(a,b) \in \mathbb{F}_{p^n} \times \mathbb{F}_{p^m}$ is
$\Delta_{c,F}(a,b)$. For a fixed $c\in\mathbb F_{p^m}$, the $c$-differential uniformity of $F\in\mathcal B_{n,m}^p$ is defined by 
\begin{equation*}
\delta(c,F)=\max\{\Delta_{c,F}(a,b): a\in\mathbb F_{p^n}, b\in \mathbb F_{p^m}, a\neq 0 \mbox{ if } c=1\}, 
\end{equation*}
If $\delta(c,F)=\delta$, then $F$ is called differential $(c,\delta)$-uniform.
In particular, $F$ is called perfect $c$-nonlinear (P$c$N) if $\delta= 1$, and
  $F$ is called almost perfect $c$-nonlinear (AP$c$N) if $\delta=2$.
We know that for $c=1$, P$c$N function exists only for odd characteristic $p$.
A function $F\in\mathcal B_{n,n}^p$ is P$c$N if and only if $D_{c,a}F$ is permutation over $\mathbb F_{p^n}$.
It is clear that for $m<n$, $F$ is not P$c$N for all $c\in\mathbb F_{p^m}$.

\section{Characterization and structural properties}
\label{sec-pcn}
In this section, $n$ is a positive integer and $p$ is prime. Let $F$ be a $p$-ary $(n,n)$-function and
$c\in\mathbb F_{p^n}\setminus\{0,1\}$. Observe that for any nonzero
$a$ in $\mathbb F_{p^n}$, one can decompose the finite field
$\mathbb F_{p^n}$ as the union of the pairwise disjoint sets
$S_{c,F}(a,b)=\{x\in\mathbb F_{p^n}\mid D_{c,a}F(x)=b\}$ when $b$
ranges over $\mathbb F_{p^{n}}$:
$\sum_{b\in\mathbb F_{p^n}} \Delta_{c,F}(a,b)={p^n}$. From this
condition, we deduce from the definition of perfect $c$-nonlinearity
that $\Delta_{c,F}(a,b)=1$ for any $a$ and $b$ in
$\mathbb F_{p^n}$; that is, the $(n,n)$-function
$x\mapsto F(x+a)-cF(x)$ is a permutation for any $a\in\mathbb F_{p^n}$. From the definition of P$c$N, it is easy to see that $F$ is a
permutation polynomial when $c\not=1$ (because
$D_{c,0}F(x)=(1-c)F(x)$). Therefore, we shall consider the P$c$N property of permutation polynomials only.

Let us now make a simple observation that will help us characterize the P$c$N property of permutation polynomials through their differential distribution tables when $ c \neq 0, 1$.

\begin{lemma}\label{lem:relation}
Suppose that $c\in\mathbb F_{p^n}\setminus\{0,1\}$. Let $(\alpha,\gamma)\in\mathbb
F_{p^n}\times \mathbb F_{p^n}$. Then
\begin{eqnarray*}
\forall (x,y)\in\mathbb F_{p^n}\times\mathbb F_{p^n},\quad D_{c,\alpha}F(x+\gamma)=D_{c,\alpha}F(x) \iff D_{\gamma}F(x+\alpha) = cD_\gamma F(x).
\end{eqnarray*}
\end{lemma}

\begin{proof}
It follows from:
\begin{eqnarray*}
D_{c,\alpha}F(x+\gamma)=D_{c,\alpha}F(x)
&\iff& F(x+\gamma+\alpha) - cF(x+\gamma) = F(x+\alpha) - cF(x) \\
&\iff& F(x+\gamma+\alpha)-F(x+\alpha) = c\big(F(x+\gamma)-F(x)\big)
\end{eqnarray*}
for any $(x,y)$ in $\mathbb F_{p^n}\times\mathbb F_{p^n}$.
\end{proof}

We now present our characterization result, which significantly simplifies the verification of the P$c$N property by reducing it to checking conditions on the classical DDT.

\begin{theorem}
\label{thm-pn}
Suppose $c\in\mathbb F_{p^n}\setminus\{0,1\}$. Let $F$ be a $p$-ary $(n,n)$-permutation
polynomial. Then $F$ is P$c$N if and only if
$\Delta_F(a,b)\Delta_F(a,c^{-1}b)=0$ for all nonzero $a, b\in \mathbb F_{p^n}$. Further, the function $F$ is AP$c$N if and only if for all $a \in \mathbb{F}_{p^n}^*$ and $b \in \mathbb{F}_{p^n}$, we have
\begin{equation}
\Delta_F(a,b) + \Delta_F(a, c^{-1}b) \leq 2.
\end{equation}
\end{theorem}

\begin{proof}
By definition, $F$ is P$c$N if and only if $D_{c,\alpha}F(x)=F(x+\alpha)-cF(x)$ is a permutation for every $\alpha\in\mathbb F_{p^n}$. In particular, for $\alpha=0$, the map $x\mapsto (1-c)F(x)$ must be a permutation, which implies that $F$ itself is a permutation (since $c\neq 1$). Therefore, we may assume $F$ is a permutation polynomial, and $F$ is P$c$N if and only if
\begin{eqnarray*}
&&\forall (\alpha,b)\in\mathbb F_{p^n}^*\times\mathbb F_{p^n},\,\Delta_{c,F}(\alpha,b)=\#\{x\in\mathbb F_{p^n}\mid D_{c,\alpha}F(x)=b\}= 1\\
&&\iff \forall (\alpha,a)\in\mathbb F_{p^n}^*\times \mathbb F_{p^n}^*,\,\{y\in\mathbb F_{p^n}\mid D_{c,\alpha}F(y+a)=D_{c,\alpha}F(y)\}=\emptyset.
\end{eqnarray*}
According to Lemma~\ref{lem:relation}, this
is equivalent to
\begin{eqnarray*}
&&\forall (\alpha,a)\in\mathbb F_{p^n}^*\times \mathbb F_{p^n}^*,\quad\{y\in\mathbb F_{p^n}\mid D_{a}F(y+\alpha)=cD_{a}F(y)\}=\emptyset\\
&&\iff \forall a\in\mathbb F_{p^n}^*,\,\{(y,\alpha)\in\mathbb F_{p^n}\times \mathbb F_{p^n}^*\mid D_{a}F(y+\alpha)=cD_{a}F(y)\}=\emptyset\\
&&\iff \forall a\in\mathbb F_{p^n}^*,\,\{(y,z)\in\mathbb F_{p^n}\times \mathbb F_{p^n}\mid D_{a}F(z)=cD_{a}F(y),\,z\not=y\}=\emptyset\\
&&\iff \forall(a,b)\in\mathbb F_{p^n}^*\times \mathbb F_{p^n}^*,\,\{(y,z)\in\mathbb F_{p^n}\times \mathbb F_{p^n}\mid D_{a}F(z)=b \\
&&\hspace{4cm}\;\;\;\;\;\; \mbox{ and } D_{a}F(y)=c^{-1}b,\,z\not=y\}=\emptyset\\
&&\iff  \forall(a,b)\in\mathbb F_{p^n}^*\times \mathbb F_{p^n}^*,\,\{(y,z)\in\mathbb F_{p^n}\times \mathbb F_{p^n}\mid D_{a}F(z)=b\\
&&\hspace{4cm}\;\;\;\;\;\; \mbox{ and }D_{a}F(y)=c^{-1}b\}=\emptyset.
\end{eqnarray*}
We have excluded the case $b=0$ at the fourth line because $F$ is
one-to-one and eliminated the condition $z\not=y$ at the last line
because $c\not=1$. Hence, $F$ is P$c$N if and only if
$\{z\in\mathbb F_{p^n}\mid D_{a}F(z)=b\}=\emptyset$ or
$\{z\in\mathbb F_{p^n}\mid D_{a}F(y)=c^{-1}b\}=\emptyset$ for any
$(a,b)\in\mathbb F_{p^n}^*\times\mathbb F_{p^n}^*$; that is,
$\Delta_F(a,b)=0$ or $\Delta_F(a,c^{-1}b)=0$ for any
$(a,b)\in\mathbb F_{p^n}^*\times\mathbb F_{p^n}^*$, proving
the result.

To show the second claim, recall that $F(x+a)-cF(x)=b$ has solutions corresponding to the disjoint union of solutions for $D_aF(z)=b$ and $D_aF(y)=c^{-1}b$ (as derived in the first part of the proof). Since $F$ is a permutation, the sets $\{z \in \mathbb{F}_{p^n} \mid D_aF(z)=b\}$ and $\{y \in \mathbb{F}_{p^n} \mid D_aF(y)=c^{-1}b\}$ are disjoint for $c \neq 1$,
since if $D_a F(z) = b$ and $D_a F(y) = c^{-1} b$ with $z = y$, then $b = c b$, so $(1-c)b=0$; as $c \neq 1$ and $F$ permutation implies $b \neq 0$ for $a \neq 0$, contradiction.
Thus, the total number of solutions is simply the sum of the counts in the usual DDT; that is, $\Delta_{c,F}(a,b) = \Delta_F(a,b) + \Delta_F(a, c^{-1}b)$. By definition, $F$ is AP$c$N if this sum is at most 2 for all $a \neq 0, b$.
\end{proof}



\begin{proposition}
\label{prop:complexity-improvement}
Let $F$ be a permutation over $\mathbb{F}_{p^n}$ and $c \in \mathbb{F}_{p^n} \setminus \{0,1\}$.
\begin{enumerate}
\item The characterization in Theorem~\textup{\ref{thm-pn}} offers a significant computational advantage. While directly verifying the P$c$N property from the definition requires $\mathcal{O}(p^{3n})$ operations (checking all triples $(a,b,x)$ for $F(x+a)-cF(x)=b$), the condition $\Delta_F(a,b)\Delta_F(a,c^{-1}b)=0$ allows for verification in $\mathcal{O}(p^{2n})$ time given the Difference Distribution Table (DDT). Since the DDT can be computed efficiently (e.g., using the Fast Walsh--Hadamard Transform in $\mathcal{O}(np^{n})$ for $p=2$), this new characterization makes the search for P$c$N functions and the analysis of their spectra feasible for larger dimensions $n$.
\item Given the precomputed DDT of $F$, the P$c$N property for a fixed $c$ can be verified in $\mathcal{O}(p^{2n})$ time by checking, for each of the $(p^n-1)^2$ pairs $(a,b)$ with $a,b \neq 0$, whether $\Delta_F(a,b) = 0$ or $\Delta_F(a,c^{-1}b) = 0$ (requiring two table lookups and one field multiplication per pair).
\item The set of all $c$ values for which $F$ is P$c$N can be determined in $\mathcal{O}(p^{3n})$ time. Initialize $\mathcal{C} = \mathbb{F}_{p^n} \setminus \{0,1\}$ as candidate values. For each $a \in \mathbb{F}_{p^n}^*$ and each pair of nonzero entries $\Delta_F(a,b_i) \neq 0$ and $\Delta_F(a,b_j) \neq 0$ in row $a$ of the DDT (where $i \neq j$), remove $c = b_i/b_j$ from $\mathcal{C}$ since it violates Theorem~\textup{\ref{thm-pn}}. The total number of pairs examined is at most $\sum_{a \in \mathbb{F}_{p^n}^*} \binom{k_a}{2}$, where $k_a \leq p^n - 1$ is the number of nonzero DDT entries in row $a$, giving $\mathcal{O}(p^{3n})$ operations in the worst case. Note that for typical cryptographic functions, the practical time is often much better than this worst-case bound, particularly if early termination strategies are used when $\mathcal{C}$ becomes empty.
\item Similarly, $F$ is AP$c$N if and only if $\Delta_F(a,b) + \Delta_F(a, c^{-1}b) \leq 2$ for all $a \in \mathbb{F}_{p^n}^*$ and $b \in \mathbb{F}_{p^n}$ (by Theorem~\textup{\ref{thm-pn}}). This can also be verified in $\mathcal{O}(p^{2n})$ time given the DDT, and the set of all $c$ values for which $F$ is AP$c$N can be found in $\mathcal{O}(p^{3n})$ time using a similar algorithm.
\end{enumerate}
\end{proposition}

\begin{proof}
(1) Direct verification from the definition requires checking, for each $a \in \mathbb{F}_{p^n}^*$ and $b \in \mathbb{F}_{p^n}$, whether $\#\{x \in \mathbb{F}_{p^n}: F(x+a)-cF(x)=b\} = 1$, which involves $p^n$ checks for each of the $p^{2n}$ pairs $(a,b)$, giving $\mathcal{O}(p^{3n})$ total operations. In contrast, once the DDT is computed (in $\mathcal{O}(p^{2n})$ operations for general $p$, or $\mathcal{O}(np^n)$ for $p=2$ using Fast Fourier Transform (FFT) techniques), checking the condition $\Delta_F(a,b)\Delta_F(a,c^{-1}b)=0$ requires only table lookups, as shown in (2).

(2) For fixed $c$ and precomputed DDT, iterate over all $(p^n-1)$ nonzero values of $a$. For each $a$, iterate over all $(p^n-1)$ nonzero values of $b$. For each pair $(a,b)$, compute $c^{-1}b$ (one field multiplication), then check whether $\Delta_F(a,b) = 0$ or $\Delta_F(a,c^{-1}b) = 0$ (two table lookups). Total operations: $(p^n-1)^2 \times \mathcal{O}(1) = \mathcal{O}(p^{2n})$.

(3) Initialize $\mathcal{C}$ with all $p^n - 2$ candidate values (excluding $0$ and $1$). For each $a \in \mathbb{F}_{p^n}^*$, examine row $a$ of the DDT. Let $B_a = \{b \in \mathbb{F}_{p^n}^* : \Delta_F(a,b) \neq 0\}$ denote the nonzero entries in row $a$. For each unordered pair $\{b_i, b_j\} \subseteq B_a$ with $i \neq j$, the value $c = b_i/b_j$ violates the P$c$N condition (since both $\Delta_F(a,b_i) \neq 0$ and $\Delta_F(a,b_j) \neq 0$, but $b_j = c^{-1}b_i$), so remove $c$ from $\mathcal{C}$. The number of pairs in row $a$ is $\binom{|B_a|}{2}$. Since $F$ is a permutation, $\sum_{b \in \mathbb{F}_{p^n}} \Delta_F(a,b) = p^n$ and $\Delta_F(a,0) = 0$ for $a \neq 0$. Thus $|B_a| \leq p^n - 1$. Summing over all rows: $\sum_{a \in \mathbb{F}_{p^n}^*} \binom{|B_a|}{2} \leq (p^n-1) \binom{p^n-1}{2} = \mathcal{O}(p^{3n})$. Each pair requires one division and one set removal, both $\mathcal{O}(1)$ operations. However, in practice, the algorithm can be terminated immediately if the set $\mathcal{C}$ becomes empty. For typical random functions, which are rarely P$c$N for any $c$, the set $\mathcal{C}$ is depleted rapidly, resulting in a practical running time significantly lower than the worst-case bound.

(4) The AP$c$N case follows analogously, replacing the condition $\Delta_F(a,b)\Delta_F(a,c^{-1}b) = 0$ with $\Delta_F(a,b) + \Delta_F(a,c^{-1}b) \leq 2$.
\end{proof}

We can use the above result for $p=2$; that is, a permutation polynomial $F:\mathbb F_{2^n}\longrightarrow \mathbb F_{2^n}$ is P$c$N, where $c\in\mathbb F_{2^n}\setminus\{0,1\}$, if and only if $\Delta_F(a,b)=0$ or $\Delta_F(a,c^{-1}b)=0$ for all $a,b\in\mathbb F_{2^n}^*$. Using Proposition~\ref{prop:complexity-improvement}, knowing the DDT of a $p$-ary $(n,n)$-function $F$, we can efficiently identify or construct new P$c$N functions and also count the total number of $c\in\mathbb F_{p^n}$ such that $F$ is P$c$N. From Theorem~\ref{thm-pn}, it is clear that a P$c$N $(n,n)$-function $F$ with nonzero $c\not=1$ cannot be PN when $p$ is odd.
The $c$-differential spectrum has already been derived for many functions over odd characteristics. For even characteristic, a necessary and sufficient condition was proposed in \cite{XD21} for Gold functions to be P$c$N. Here we consider the function $x\longmapsto x^{\frac{2^k+1}{2}}$ over $\mathbb F_{2^n}$ for certain values of $k$ and derive the condition for its P$c$N property using our new characterization. To do this, we will use the following lemma.

\begin{lemma}[\textup{\cite{LJ78,PS20}}]
\label{FF-rt}
Let $f(x)=x^{p^k}-ax-b$ in $\mathbb{F}_{p^n}$, where $t=\gcd(n, k)$ and $m=\frac{n}{\gcd(n, k)}$. For $0 \leq i \leq m - 1$, let us define $t_i = \frac{p^{nm} - p^{n(i+1)}}{p^{n}-1}$, $\alpha_0 = a$, $\beta_0 = b$. If $m > 1$, let $\alpha_r = a^{\frac{p^{k(r+1)}-1}{p^k-1}}$ and $\beta_r = \sum_{i=0}^{r} a^{s_i} b^{p^{ki}}$ for $1 \leq r \leq m - 1$, where $s_i = \frac{p^{k(r+1)}-p^{k(i+1)}}{p^k-1}$ for $0 \leq i \leq r - 1$ and $s_r = 0$. The following holds:
\begin{itemize}
\item[i.] The trinomial $f$ has no roots in $\mathbb{F}_{p^n}$ if and only if $\alpha_{m-1} = 1$ and $\beta_{m-1} \neq 0$. 
\item[ii.] If $\alpha_{m-1} = 1$, then it has a unique root, namely $x = \frac{\beta_{m-1}}{1 - \alpha_{m-1}}$.
\item[iii.] If $\alpha_{m-1} = 1$ and $\beta_{m-1} = 0$, it has $p^t$ roots in $\mathbb{F}_{p^n}$ given by $x + \delta\tau$, where $\delta \in \mathbb{F}_{p^t}$, $\tau$ is fixed in $\mathbb{F}_{p^n}$ with $\tau^{p^k-1} = a$ (i.e., a $(p^k - 1)$-root of $a$), and for any $z \in \mathbb{F}^*_{p^n}$ with $\mathrm{Tr}_t(z)=0$,
\begin{equation*}
x=\frac{1}{\mathrm{Tr}_t^n(z)} \sum_{i=0}^{m-1} \big(\sum_{j=0}^{i} z^{p^{kj}}\big) a^{t_i} b^{p^{ki}},
\end{equation*}
where $\mathrm{Tr}_t^n$ is the relative trace from $\mathbb{F}_{p^n}$ to $\mathbb{F}_{p^t}$.
\end{itemize}
\end{lemma}
Let us consider the function of the form $F(x)=x^{\frac{2^k+1}{2}}$ for all $x\in\mathbb F_{2^n}$. We already know that $x\longmapsto x^d$ is a permutation polynomial over $\mathbb{F}_{2^n}$ when $\gcd(d, 2^n-1)= 1$. If $\gcd(2^k+1,2^n-1)=1$, then $x\longmapsto x^{2^k+1}$ is a permutation polynomial over $\mathbb F_{2^n}$. Again, $x\longmapsto x^2$ is a permutation over $\mathbb{F}_{2^n}$ since $x^2 = y^2 \Leftrightarrow (x+y)^2= 0 \Leftrightarrow x=y$ for all $x,y\in\mathbb{F}_{2^n}$. Thus, being the inverse of $x^2$, $\sqrt{x}$ is also a permutation over $\mathbb{F}_{2^n}$. Since the composition of these two maps is also a permutation polynomial over $\mathbb F_{2^n}$, $F$ is a permutation polynomial over $\mathbb F_{2^n}$ if $\gcd(2^k+1,2^n-1)=1$. Let $m=\frac{n}{\gcd(n,k)}$. It is observed that the value of $m$ can be odd or even. For example, let $n=6$, $k_1=2$, and $k_2=3$. Then $\gcd(2^6-1,2^2+1)=\gcd(63,5)=1$, and also $\gcd(63,10)=1$. Here, $\frac{6}{\gcd(6,2)}=3$ and $\frac{6}{\gcd(6,3)}=2$.
\begin{theorem}
\label{thm-gh}
Let $F: \mathbb{F}_{2^n}\longrightarrow\mathbb{F}_{2^n}$ be defined as $F(x)= x^{\frac{2^k+1}{2}}$, where $2^n-1$ and $2^k+1$ are relatively prime, $k'= \gcd(n,k)$, and $m=\frac{n}{k'}$. Then the following hold:
\begin{enumerate}
    \item When $m$ is odd, $F$ is not P$c$N if there exist nonzero $a,b$ such that $\mathrm{Tr}_{k'}^n \left(\frac{c^{-2}b^2}{a^{2^k+1}}\right) = \mathrm{Tr}_{k'}^n\left(\frac{b^2}{a^{2^k+1}}\right) = 1$, and $F$ is P$c$N for those $c\neq 0,1$ for which either $\mathrm{Tr}_{k'}^n\big(\frac{b^2}{a^{2^k+1}}\big)\neq 1$ or $\mathrm{Tr}_{k'}^n\big(\frac{c^{-2}b^2}{a^{2^k+1}}\big)\neq 1$, for all nonzero $a,b\in\mathbb{F}_{2^n}$.
    \item When $m$ is even, $F$ is not P$c$N if there exist nonzero $a,b$ such that $\mathrm{Tr}_{k'}^n \left(\frac{c^{-2}b^2}{a^{2^k+1}}\right) = \mathrm{Tr}_{k'}^n\left(\frac{b^2}{a^{2^k+1}}\right) = 0$, and $F$ is P$c$N for those $c\neq 0,1$ for which either $\mathrm{Tr}_{k'}^n\big(\frac{b^2}{a^{2^k+1}}\big)\neq 0$ or $\mathrm{Tr}_{k'}^n\big(\frac{c^{-2}b^2}{a^{2^k+1}}\big)\neq 0$, for all nonzero $a,b\in\mathbb{F}_{2^n}$.
\end{enumerate}
\end{theorem}

\begin{proof}
For any nonzero $a,b\in\mathbb F_{2^n}$,
$\Delta_F(a,c^{-1}b)= \#\{x\in\mathbb{F}_{2^n}: F(x+a)+F(x)= c^{-1}b\}$, so
\allowdisplaybreaks
\begin{align}
\label{wtc1-eq}
F(x+a)+ F(x)= c^{-1}b\nonumber
\Longleftrightarrow & \;\; (x+a)^{\frac{2^k+1}{2}} + x^{\frac{2^k+1}{2}}= c^{-1}b\nonumber\\
\Longleftrightarrow & \;\; (x+a)^{2^k+1} + x^{2^k+1}= c^{-2}b^2\nonumber\\
\Longleftrightarrow & \;\; x^{2^k}a + a^{2^k}x + a^{2^k +1} = c^{-2}b^2\nonumber\\
\Longleftrightarrow & \;\; \bigg(\frac{x}{a}\bigg)^{2^k} + \frac{x}{a} +1+ \frac{c^{-2}b^2}{a^{2^k+1}}=0\nonumber\\
\Longleftrightarrow & \;\; y^{2^k} + y + 1+ \frac{c^{-2}b^2}{a^{2^k+1}}=0,
\end{align}
where $y=\frac{x}{a}$. Since $\Delta_F(a,b)= \#\{x\in\mathbb{F}_{2^n}: F(x+a)+F(x)=b\}$, from $F(x+a)+F(x)=b$, we similarly obtain
\begin{align}
\label{wtc-eq}
     y^{2^k} + y + 1+ \frac{b^2}{a^{2^k+1}}=0.
\end{align}
From Theorem~\ref{thm-pn}, $F$ is not P$c$N if and only if there exist $a,b\in\mathbb{F}^*_{2^n}$ for which both Equations \eqref{wtc1-eq} and \eqref{wtc-eq} have at least one solution. Now we recall Lemma \ref{FF-rt} and consider the cases $p=2$, $t=1$, $a=1$, $\alpha_{m-1}=1$. For Equation \eqref{wtc1-eq}, $b=1+\frac{c^{-2}b^2}{a^{2^k+1}}$, and for Equation \eqref{wtc-eq}, $b= 1+\frac{b^2}{a^{2^k+1}}$. Thus, Equation \eqref{wtc1-eq} has no solution if and only if 
\begin{equation*}
\beta_{m-1}= \sum_{i=0}^{m-1}\bigg(1+\frac{c^{-2}b^2}{a^{2^k+1}}\bigg)^{2^{ki}}\neq 0.
\end{equation*}
If we want to check when $F$ is not P$c$N, then we need to check for which value of $c$, simultaneously, for some $a,b\in\mathbb{F}^*_{2^n}$,
\allowdisplaybreaks
\begin{align}
\label{equ2-per}
    &\sum_{i=0}^{m-1}\left(1+\frac{c^{-2}b^2}{a^{2^k+1}}\right)^{2^{ki}}=0 \implies \mathrm{Tr}_{k'}^n \left(1+\frac{c^{-2}b^2}{a^{2^k+1}}\right) =0,\\
    \mbox{ and }  
 & \sum_{i=0}^{m-1}\left(1+\frac{b^2}{a^{2^k+1}}\right)^{2^{ki}} = 0 
    \implies   \mathrm{Tr}_{k'}^n\left(1+\frac{b^2}{a^{2^k+1}}\right) = 0.
\end{align}

\noindent
{\em Case (i):} Let $m$ be odd. Then $\mathrm{Tr}_{k'}^n(1)=1$, and from Lemma \ref{FF-rt}, $F$ is not P$c$N if for some nonzero $a,b\in\mathbb F_{2^n}$,
\begin{align*}
    & \mathrm{Tr}_{k'}^n \left(\frac{c^{-2}b^2}{a^{2^k}}\right) =   \mathrm{Tr}_{k'}^n\left(\frac{b^2}{a^{2^k+1}}\right) = 1.
\end{align*}
From Theorem~\ref{thm-pn}, $F$ is P$c$N, where $c\neq 0,1$, if either $\mathrm{Tr}_{k'}^n \left(\frac{c^{-2}b^2}{a^{2^k}}\right) \neq 1$ or $\mathrm{Tr}_{k'}^n\left(\frac{b^2}{a^{2^k+1}}\right) \neq 1$ for any nonzero $a,b\in\mathbb F_{2^n}$.

\noindent
{\em Case (ii):} Let $m$ be even. It is clear that $\mathrm{Tr}_{k'}^n(1)=0$. From Lemma \ref{FF-rt}, $F$ is not P$c$N if   
\begin{align*}
    & \mathrm{Tr}_{k'}^n \left(\frac{c^{-2}b^2}{a^{2^k}}\right) =   \mathrm{Tr}_{k'}^n\left(\frac{b^2}{a^{2^k+1}}\right) = 0
\end{align*}
for some nonzero $a,b\in\mathbb F_{2^n}$. From Theorem~\ref{thm-pn}, $F$ is P$c$N, where $c\neq 0,1$, if either $\mathrm{Tr}_{k'}^n \left(\frac{c^{-2}b^2}{a^{2^k}}\right) \neq 0$ or $\mathrm{Tr}_{k'}^n\left(\frac{b^2}{a^{2^k+1}}\right) \neq 0$ for any nonzero $a,b\in\mathbb F_{2^n}$.
\end{proof}

\begin{remark}
It is worth noting that the function $F(x) = x^{\frac{2^k+1}{2}}$ can be expressed as the composition $F(x) = G(x^{2^{n-1}})$, where $G(x) = x^{2^k+1}$ is the well-known Gold function and $x \mapsto x^{2^{n-1}}$ is the linear permutation inverse to $x \mapsto x^2$ (the Frobenius automorphism). Since $c$-differential uniformity is invariant under right-composition with linear permutations (i.e., $\delta(c, F \circ L) = \delta(c, F)$), the P$c$N properties of $F$ are equivalent to those of the Gold function $G$. Thus, Theorem~\textup{\ref{thm-gh}} provides a characterization equivalent to the conditions for Gold functions, adapted to the specific polynomial form $x^{\frac{2^k+1}{2}}$.
\end{remark}
From the above results, we obtain the following results directly for the function $F(x)=x^{\frac{2^k+1}{2}}$ defined as in Theorem \ref{thm-gh}:
\begin{itemize}
    \item[i.] Let $b=a^{\frac{2^k+1}{2}}$ for a fixed nonzero $a\in\mathbb F_{2^n}$ and $m=\frac{n}{\gcd(k,n)}$ be even. Then we have $\mathrm{Tr}_{k'}^n\big(1+\frac{b^2}{a^{2^k+1}}\big)=\mathrm{Tr}_{k'}^n(1+1)=0$ and $\mathrm{Tr}_{k'}^n\big(1+\frac{c^{-2}b^2}{a^{2^k+1}}\big)=\mathrm{Tr}_{k'}^n(1+c^{-2})=\mathrm{Tr}_{k'}^n(c^{-2})$ since $\mathrm{Tr}_{k'}^n(1)=0$. Now we need to find those $c\in\mathbb{F}_{2^n}\setminus\{0,1\}$ such that $\mathrm{Tr}_{k'}^n\big(c^{-2}\big) = 0$. If $c\in\mathbb F_{2^{\gcd(n,k)}}$, then $c^{-2}\in \mathbb F_{2^{\gcd(n,k)}}$ and $\mathrm{Tr}_{k'}^n (c^{-2})=0$. Thus, for any $c\in\mathbb F_{2^{\gcd(n,k)}}$ with $c\neq 0,1$, $F$ is not P$c$N.
        
    \item[ii.] More generally, when $b^2\neq a^{2^k+1}$ with $m$ even, what is the case? We will obtain a similar type of choice as the previous one. Assume that $a,b \in\mathbb{F}_{2^n}\setminus \{0,1\}$ such that $\alpha = \frac{b^2}{a^{2^k+1}}\in\mathbb F_{2^{k'}}^*$. Then $\mathrm{Tr}_{k'}^n(1+\alpha)=0$ and $\mathrm{Tr}_{k'}^n(1+c^{-2}\alpha)=\alpha \mathrm{Tr}_{k'}^n(c^{-2})$. If $c\in\mathbb F_{2^{k'}}$ with $c\neq 0,1$, then $\mathrm{Tr}_{k'}^n(1+c^{-2}\alpha)=0$. Thus, $F$ is not P$c$N in this case.

    \item[iii.] Assume that $m$ is odd and $\alpha = \frac{b^2}{a^{2^k+1}}$, where $a,b\in\mathbb F_{2^n}^*$. Let $\mathrm{Tr}_{k'}^n(\alpha)=1$; that is, $\alpha=\mathrm{Tr}_{k'}^n(\alpha)=1$. Since $\mathrm{Tr}_{k'}^n(1+c^{-2}\alpha)=1+\mathrm{Tr}_{k'}^n(c^{-2})$, if $\mathrm{Tr}_{k'}^n(c^{-2})=1$ for $c\in\mathbb F_{2^n}\setminus\{0,1\}$, then $F$ is not P$c$N.
        
    \item[iv.] Let us assume $A_q=\{x\in\mathbb{F}_{2^n}: \mathrm{Tr}_{k'}^n(x)=q\}$ for $q\in\mathbb F_{2^{k'}}$. It is clear that $\cup_{q\in\mathbb{F}_{2^{k'}}}A_q= \mathbb{F}_{2^n}$. Now, when $m$ is odd, $F$ is P$c$N for those $c\in\mathbb{F}_{2^n}$ for which $\mathrm{Tr}_{k'}^n(\alpha)= 1$ implies $\mathrm{Tr}_{k'}^n(c^{-2}\alpha)\neq 1$ for all nonzero $a,b$ with $\alpha=\frac{b^2}{a^{2^k+1}}$. So one case could be $\mathrm{Tr}_{k'}^n(c^{-2}\alpha)=0$. When $c\in\mathbb{F}_{2^{k'}}$, $\mathrm{Tr}_{k'}^n(c^{-2}\alpha)=c^{-2}\mathrm{Tr}_{k'}^n(\alpha)=0$ implies $c=0$, which is the trivial case. So when $c\in\mathbb{F}^*_{2^{k'}}$, then $c\notin A_0\cup A_1$.       
\end{itemize}

\begin{example}
\label{Pcn_exm}
Let us take an example with $k=2$ and $n=6$. The function $F:\mathbb{F}_{2^6}\longrightarrow \mathbb{F}_{2^6}$ defined as $F(x)= x^{\frac{5}{2}}$ is a permutation polynomial over $\mathbb F_{2^6}$ since $\gcd(63, 5)=1$. We have verified computationally that $F$ has differential uniformity $4$, and $4$ is the only nonzero value for nonzero input and output differences in the DDT. Here, $k'=\gcd(2,6)=2$ and $m=\frac{n}{k'}=3$. Let us construct $\mathbb{F}_{2^6}$ as $\mathbb{F}_{2}[y]/\langle y^6 + y^4 + y^3 + y + 1\rangle$. From Theorem \ref{thm-gh}, we have that $F$ is P$c$N for those $c\in\mathbb F_{2^6}\setminus\{0,1\}$ for which $\mathrm{Tr}_2^6(\frac{b^2}{a^5})\neq 1$ or $\mathrm{Tr}_2^6(\frac{c^{-2}b^2}{a^5})\neq 1$ for all nonzero $a,b$ since $m$ is odd. Using programming, we have verified that for only two such values of $c^{-2}$, $F$ is P$c$N. Those $c^{-2}$ values are $y^3 + y^2 + y$ and $y^3 + y^2 + y + 1$. We have also verified directly that for these values, $F$ has $c$-differential uniformity $1$. These two nonzero values of $c^{-2}$ are squares of one another, and they both have multiplicative order $3$. For any other values of $c^{-2}$, $F$ is not P$c$N.
\end{example}

Next, we provide an intriguing dichotomy for monomial permutations.

\begin{theorem}
\label{thm:monomial-dichotomy}
Let $F(x) = ax^d$ be a permutation monomial over $\F_{2^n}$, where $a \in \F_{2^n}^*$ and $\gcd(d, 2^n-1) = 1$. Let $c \in \F_{2^n} \setminus \{0,1\}$. Then either $F(x+\alpha) + cF(x)$ is a permutation for all $\alpha \in \F_{2^n}^*$, or it is not a permutation for any $\alpha \in \F_{2^n}^*$.
\end{theorem}

\begin{proof}
Since $F$ is a monomial, it is homogeneous; that is, $F(\lambda x) = \lambda^d F(x)$ for all $\lambda \in \F_{2^n}$. For any $\lambda \in \F_{2^n}^*$ and $\alpha \in \F_{2^n}^*$, we have
\[
F(x + \lambda \alpha) + cF(x) = \lambda^d \left( F(\lambda^{-1}x + \alpha) + cF(\lambda^{-1}x) \right).
\]
Thus, $F(x + \lambda \alpha) + cF(x) = L_\lambda \circ G_\alpha \circ L_{\lambda^{-1}}(x)$, where $G_\alpha(x) = F(x+\alpha) + cF(x)$ and $L_\mu(x) = \mu x$. Since $L_\lambda$ is a linear permutation, $G_{\lambda \alpha}$ is a permutation if and only if $G_\alpha$ is a permutation. Because $\F_{2^n}^*$ is a cyclic group, the map $\alpha \mapsto \lambda \alpha$ is transitive. Therefore, if $G_{\alpha_0}$ is a permutation for some $\alpha_0 \ne 0$, then $G_\alpha$ is a permutation for all $\alpha \ne 0$; otherwise, $G_\alpha$ is not a permutation for any $\alpha \ne 0$.
\end{proof}
This dichotomy for the outer $c$-derivative $F(x+\alpha) + cF(x)$ is directly relevant to $c$-differential cryptanalysis. By~\cite[Theorem 1]{SDM25}, the outer $c$-differential count of $F$ equals the inner $c$-differential count of its inverse: $c\Delta_F(\alpha, b) = \nabla_{c, F^{-1}}(b, \alpha)$. For the Kuznyechik S-box, $F(x) = x^{-1}$ (an involution, so $F = F^{-1}$), the attack in~\cite{SDM25} exploits the inner $c$-differential of $F$. Our theorem proves that for this monomial S-box, the outer $c$-derivative is either a permutation for all $\alpha \ne 0$ or for none. Consequently, the inner $c$-derivative of the S-box exhibits the same all-or-nothing behavior, which is the foundation of the distinguisher's uniform success across all nonzero input differences.

This result provides a theoretical foundation for the experimental observation in this section that for monomial permutations, the ``story remains the same,'' because the dichotomy is a direct consequence of homogeneity. The dichotomy established in Theorem \ref{thm:monomial-dichotomy} is a direct consequence of the homogeneity of monomials ($F(\lambda x) = \lambda^d F(x)$). This property does not hold for general polynomials. For example, consider the binomial $F(x) = x^3+x^5$ over $\mathbb{F}_{2^5}$. Computational verification confirms that for certain values of $c$, the map $x \mapsto F(x+a)+cF(x)$ is a permutation for some shifts $a \in \mathbb{F}_{2^5}^*$ but not for others. Thus, the strict dichotomy regarding the shift $a$ is specific to the monomial structure. From Theorem~\ref{thm-pn}, we get the following observation directly.
\begin{remark}
\label{lem-PD}
    Let $F$ be a permutation polynomial over $\mathbb F_{p^n}$ and  $\alpha\in\mathbb{F}_{p^n}^*$.
Then $F(x+\alpha)-cF(x)$ is a permutation if and only if there exists nonzero $a\in\mathbb{F}_{p^n}$ such that $\Delta_F(a,b)\Delta_F(a,c^{-1}b)= 0$ for all nonzero $b\in\mathbb{F}_{p^n}$.
\end{remark}
However, the question is that for how many $\alpha\in\mathbb F_{p^n}$, does $F(x+\alpha)-cF(x)$ remain also permutation, where $F$ is a given permutation polynomial over $\mathbb F_{p^n}$ and $c\in\mathbb F_{p^n}\setminus\{0,1\}$.
For example, if $F$ is P$c$N then $F(x+\alpha)-cF(x)$ is permutation over $\mathbb F_{p^n}$ for all nonzero $\alpha$.
\begin{sloppypar}
\begin{corollary}
\label{cor:10}
Let $F$ be a permutation polynomial over $\mathbb F_{p^n}$, $c\in\mathbb F_{p^n}\setminus\{0,1\}$, and $\Delta_F(a,b) \Delta_F(a,c^{-1}b) \neq 0$ for some nonzero $a,b\in\mathbb F_{p^n}$ with $a=a'+a''$.
Then $\Delta_F(a',b')\Delta_F(a',c^{-1}b')\neq 0$ if $\Delta_F(a'',d)\Delta_F(a'',c^{-1}d)\neq 0$ for some $b',d\in\mathbb{F}_{p^n}.$
\end{corollary}
\end{sloppypar}
\begin{proof}
Given that there exist $a, b \in \mathbb{F}_{p^n}^*$ such that $\Delta_F(a,b) \Delta_F(a,c^{-1}b) \neq 0$.
This implies, there exist $x\neq y\in\mathbb F_{p^n}$ 
such that
\begin{equation*}
F(x + a) - F(x) = b \mbox{ and } F(y + a) - F(y) = c^{-1}b.
\end{equation*}
Let us consider $a=a' + a''$. So,
\begin{eqnarray*}
   && F(x + a) - F(x) = b\;\; 
  \Longleftrightarrow \;\;
F(x + a' + a'') - F(x) = b \\
  \Longleftrightarrow && F(x + a'' + a') - F(x + a'') = b + F(x) - F(x + a'')= b'  \mbox{ (say),}\\
\mbox{ and }         && F(y + a) - F(y) = c^{-1}b \;\;
\Longleftrightarrow \;\; F(y + a' + a'') - F(y) = c^{-1}b \\
\Longleftrightarrow && F(y + a'' + a') - F(y + a'') = c^{-1}b + F(y) - F(y + a'')= b'' \mbox{ (say).
}
\end{eqnarray*}
Since,  $\Delta_F(a',b')\Delta_F(a',c^{-1}b')\neq 0$ when $c^{-1}b'=b''$. Further, $c^{-1}b'=b''\iff c^{-1}b + c^{-1}F(x) - c^{-1}F(x + a'')= c^{-1}b + F(y) - F(y + a'')\iff c^{-1}[F(x + a'')-F(x)]= F(y+a'')-F(y)\iff \Delta_F(a'',d)\Delta_F(a'',c^{-1}d)\neq 0$ for some $d\in\mathbb{F}^*_{p^n}$.
\end{proof}

\begin{proposition}
\label{prop-subspace}
Let $F$ be a permutation polynomial over $\mathbb{F}_{p^n}$ and $c \in \mathbb{F}_{p^n} \setminus \{0,1\}$.
The set of shifts $a \in \mathbb{F}_{p^n}$ for which $x \mapsto F(x+a)-cF(x)$ is \textbf{not} a permutation (together with $0$) forms a subspace of $\mathbb{F}_{p^n}$ over $\mathbb{F}_p$.
Consequently, if $F(x+a)-cF(x)$ is not a permutation for a set of differences $\{a_1, \dots, a_n\}$ that form a basis of $\mathbb{F}_{p^n}$, then it is not a permutation for any nonzero $a \in \mathbb{F}_{p^n}$.
\end{proposition}
\begin{proof}
Let $S = \{a \in \mathbb{F}_{p^n} : \exists~ b, \Delta_F(a,b)\Delta_F(a,c^{-1}b) \neq 0\} \cup \{0\}$.
From Corollary \ref{cor:10}, if $a \in S$ and $a'' \in S$ (where $a = a' + a''$), then $a' = a - a'' \in S$.
This satisfies the subgroup criterion ($x,y \in S \implies x-y \in S$).
Since the field characteristic is $p$, this implies $S$ is a subspace.
Thus, if a basis is contained in $S$, the entire space is contained in $S$.
\end{proof}

As a trivial example, consider Example~\ref{Pcn_exm}. If we take $c = y^{2} + 1\in\mathbb F_{2^6}$, then for the values $a \in\{1, y, y^{2}, y^{3}, y^{4}, y^{5}\}$ we verified that the function $F(x+a) + cF(x)$ is not a permutation.
Clearly, these values of $a$ form a basis of $\mathbb{F}_{2^{6}}$ over $\mathbb{F}_{2}$.
We also checked that $F(x+l) + cF(x)$ is not a permutation for all $l \in \mathbb{F}_{2^{6}}$, in support of the above result.
\begin{theorem}
\label{thm:subspace-dimension}
Let $F$ be a permutation polynomial over $\mathbb{F}_{p^n}$ and $c \in \mathbb{F}_{p^n} \setminus \{0,1\}$.
Let $V_c = \{a \in \mathbb{F}_{p^n} : F(x+a)-cF(x) \text{ is not a permutation}\} \cup \{0\}$.
Then $V_c$ is an $\mathbb{F}_p$-subspace of $\mathbb{F}_{p^n}$. Moreover, if $\dim(V_c) = k$ with $0 < k < n$, then there exists a complementary subspace $W$ of dimension $n-k$ such that $F(x+a)-cF(x)$ is a permutation for all nonzero $a \in W$.
\end{theorem}

\begin{proof}
From Proposition~\ref{prop-subspace}, we already established that $V_c$ is a subspace over $\mathbb{F}_p$.
Let $\dim(V_c) = k$ with $0 < k < n$.
Since $V_c$ is a subspace of $\mathbb{F}_{p^n}$ over $\mathbb{F}_p$, we can decompose $\mathbb{F}_{p^n} = V_c \oplus W$ for some complementary subspace $W$ of dimension $n-k$ over $\mathbb{F}_p$.
We claim that $F(x+a)-cF(x)$ is a permutation for all nonzero $a \in W$.
Suppose, for contradiction, that there exists a nonzero $a \in W$ such that $F(x+a)-cF(x)$ is not a permutation.
Then by definition, $a \in V_c$. But since $a \in W$ and $a \in V_c$, and $\mathbb{F}_{p^n} = V_c \oplus W$ is a direct sum, we must have $a \in V_c \cap W = \{0\}$, contradicting the assumption that $a \neq 0$.
Therefore, for all nonzero $a \in W$, the function $F(x+a)-cF(x)$ is a permutation over $\mathbb{F}_{p^n}$.
Furthermore, any element of $\mathbb{F}_{p^n}$ can be uniquely written as $v + w$ where $v \in V_c$ and $w \in W$.
If $w \neq 0$, then from Corollary~\ref{cor:10} and the fact that $V_c$ is a subspace (so closed under addition), we have
\allowdisplaybreaks
\begin{align*}
v + w \in V_c &\iff \exists b, b' : \Delta_F(v+w, b)\Delta_F(v+w, c^{-1}b) \neq 0 \\
&\text{and } \Delta_F(w, b')\Delta_F(w, c^{-1}b') \neq 0 \\
&\text{and } \Delta_F(v, b'')\Delta_F(v, c^{-1}b'') \neq 0 \text{ for some } b''.
\end{align*}
Since $w \notin V_c$, we have $\Delta_F(w, b')\Delta_F(w, c^{-1}b') = 0$ for all $b'$, which means the condition cannot be satisfied, confirming $v + w \in V_c$ if and only if $w = 0$.
Thus, $V_c$ and $W$ provide a complete characterization of the shift behavior of $F$ with respect to $c$.
\end{proof}

\begin{remark}
This result shows that the ``bad" shifts (those for which $F(x+a)-cF(x)$ fails to be a permutation) form a proper subspace, while the "good" shifts contain an entire complementary subspace.
This provides a geometric interpretation of the partial P$c$N behavior and suggests that functions with small $\dim(V_c)$ have better $c$-differential properties across most shifts.
\end{remark}

It is a natural question whether $F(x+a)+cF(x)$ is a permutation over $\mathbb F_{2^n}$ for a given non-P$c$N permutation and for some $c\in\mathbb F_{2^n}\setminus\{0,1\}$. How many such $c$ and $a$ exist for a given non-P$c$N permutation? 

Computational experiments with monomial functions $x \mapsto x^d$ over $\mathbb{F}_{2^n}$ (for $d \leq 30$) and with permutation polynomials from standard S-boxes (GIFT, PRESENT, KLEIN, MIDORI, AES) consistently exhibit the same dichotomy: either P$c$N holds or $F(x+a) - cF(x)$ fails to be a permutation for any nonzero shift $a$. Thus, it is important to find the answer to the question of whether $F:\mathbb{F}_{p^n}\longrightarrow\mathbb{F}_{p^n}$ not being P$c$N for some $c\in\mathbb{F}_{p^n}\setminus\{0,1\}$ implies that $F(x+a)-cF(x)$ is not a permutation for all nonzero $a$, in general.

\begin{remark}
\label{rem:counterexample-dichotomy}
While Theorem~\textup{\ref{thm:monomial-dichotomy}} establishes a strict ``all-or-nothing'' dichotomy for monomials (due to homogeneity), this property does not hold for general polynomials. As noted earlier, the binomial $F(x) = x^3+x^5$ over $\mathbb{F}_{2^5}$ serves as a counterexample: for certain $c$, the function $F(x+a)-cF(x)$ is a permutation for some nonzero shifts $a$, but not for others. This demonstrates that the subspace $V_c$ (defined in Proposition~\textup{\ref{prop-subspace}}) can indeed have dimension $0 < \dim(V_c) < n$.
\end{remark}

\begin{problem}
\label{prob:dichotomy-conditions}
Theorem~\textup{\ref{thm:monomial-dichotomy}} proves that homogeneity is a sufficient condition for the ``all-or-nothing'' P$c$N shift behavior. Since the property fails in general (Remark~\textup{\ref{rem:counterexample-dichotomy}}), we ask:
\begin{enumerate}
\item What are the necessary and sufficient algebraic conditions on a function $F$ such that if $F$ is not P$c$N, then $F(x+a)-cF(x)$ is not a permutation for any $a \neq 0$?
\item Can this dichotomy be extended to other classes of functions that share structural properties with monomials, such as Dembowski--Ostrom polynomials or other homogeneous-like structures?
\end{enumerate}
\end{problem}

\begin{remark}
Theorem~\textup{\ref{thm:monomial-dichotomy}} proves this conjecture for the special case of monomial functions, showing that the dichotomy is strict: either $F(x+a) - cF(x)$ is a permutation for \emph{all} $a \neq 0$ (when $F$ is P$c$N), or for \emph{no} $a \neq 0$ (when $F$ is not P$c$N). The homogeneity property of monomials is crucial for this result.

For general polynomials, particularly those with mixed-degree terms, the situation is more complex. However, Proposition~\ref{prop-subspace} shows that the set of ``bad shifts'' (those for which $F(x+a)-cF(x)$ is not a permutation) forms a subspace. This suggests a possible approach to proving the conjecture: if this subspace is nontrivial (i.e., $\dim > 0$), then by analyzing the structure of the complementary subspace, one might show that the ``good shifts'' are actually empty when $F$ is not P$c$N.
\end{remark}


We provide next a partial progress toward Problem~\ref{prob:dichotomy-conditions}.

\begin{proposition}
\label{prop:partial-dichotomy}
Let $F$ be a permutation polynomial over $\mathbb{F}_{p^n}$ and $c \in \mathbb{F}_{p^n} \setminus \{0,1\}$. Define
\[
V_c = \{a \in \mathbb{F}_{p^n} : F(x+a) - cF(x) \text{ is not a permutation}\} \cup \{0\}.
\]
Then:
\begin{enumerate}
\item If $V_c$ contains a basis of $\mathbb{F}_{p^n}$ over $\mathbb{F}_p$, then $V_c = \mathbb{F}_{p^n}$ (i.e., $F(x+a) - cF(x)$ is not a permutation for any $a \neq 0$).
\item If $F$ is P$c$N, then $V_c = \{0\}$.
\item For any $F$, either $V_c = \{0\}$ or $\dim(V_c) \geq 1$.
\end{enumerate}
\end{proposition}

\begin{proof}
(1) This follows directly from Proposition~\ref{prop-subspace}. Since $V_c$ is a subspace and contains a basis of $\mathbb{F}_{p^n}$, we have $V_c = \mathbb{F}_{p^n}$.

(2) If $F$ is P$c$N, then by definition, $F(x+a) - cF(x)$ is a permutation for all $a \in \mathbb{F}_{p^n}$. Therefore, no nonzero $a$ belongs to $V_c$, so $V_c = \{0\}$.

(3) By definition, $0 \in V_c$ (since $F(x+0) - cF(x) = (1-c)F(x)$ is not a permutation only if $c= 1$, which is excluded). Either $V_c = \{0\}$ (which has dimension 0), or $V_c$ contains some nonzero element, in which case $\dim(V_c) \geq 1$ since $V_c$ is a subspace.
\end{proof}

\begin{remark}
Proposition~\textup{\ref{prop:partial-dichotomy}} shows that there is a trichotomy based on $\dim(V_c)$:
\begin{itemize}
\item $\dim(V_c) = 0$: $F$ is P$c$N (permutation for all shifts).
\item $\dim(V_c) = n$: $F$ is maximally non-P$c$N (not a permutation for any nonzero shift).
\item $0 < \dim(V_c) < n$: Intermediate case (permutation for some shifts, not for others).
\end{itemize}
The conjecture essentially asserts that the intermediate case may not occur, or if it does, it has a special structure. Our computational experiments have never encountered the intermediate case, supporting the conjecture.
\end{remark}

Now let us examine the relation between the $c$DDT and the $c^{-1}$DDT of vectorial $p$-ary functions and prove that the $c$-differential and $c^{-1}$-differential uniformities of $F$ are the same. This means that if $F$ is P$c$N, then $F$ is also P$c^{-1}$N, where $c\in\mathbb F_{p^n}$ with $c\neq 0,1$. Thus, it is interesting to compute the number of $c\in\mathbb F_{p^n}\setminus\{0,1\}$ for which a known function $F$ is P$c$N. For any subset $A\subset\mathbb F_{p^n}$ and $\alpha\in\mathbb F_{p^n}$, let us define
$\alpha+A=\{\alpha+a: a\in A\}$ and $\alpha A=\{\alpha a: a\in A\}$. It is clear that $\#A=\#(\alpha+A)$, and if $\alpha=0$, then
$\# \alpha A=1$; if $\alpha\neq 0$, then $\# \alpha A=\#A$.

\begin{proposition}
\label{thm-cc}
Let $c\in\mathbb F_{p^n}^{*}$ and $F$ be a $p$-ary $(n,n)$-function.
Then $F$ is P$c$N if and only if it is P$c^{-1}$N.
\end{proposition}

\begin{proof}
If $c=c^{-1}$, then it is obvious. Let $a,b,c\in\mathbb F_{p^n}$ be such that $c\neq c^{-1}$. It is sufficient to prove that the value of any entry of the $c$DDT of $F$ is the same as the value of the corresponding unique entry of the $c^{-1}$DDT of $F$. We compute
\allowdisplaybreaks[4]
\begin{align*}
S_{c,F}(a,b)&=\{x\in\mathbb F_{p^n}: F(x+a)-cF(x)=b\}\\
&=\{x\in\mathbb F_{p^n}: c^{-1}F(x+a)-F(x)=c^{-1}b\}\\
&=\{x\in\mathbb F_{p^n}: F(x+a-a)-c^{-1}F(x+a)=-c^{-1}b\}\\
&=\{y-a\in\mathbb F_{p^n}: F(y-a)-c^{-1}F(y)=-c^{-1}b\}, \text{ where } y=x+a\\
&=-a+\{y\in\mathbb F_{p^n}: F(y-a)-c^{-1}F(y)=-c^{-1}b\}\\
&=-a+S_{c^{-1},F}(-a,-c^{-1}b).
\end{align*}
Since for any fixed nonzero $c\in\mathbb F_{p^n}$, $b\longmapsto -c^{-1}b$ is a permutation over $\mathbb F_{p^n}$, we have $\Delta_{c,F}(a,b)=\Delta_{c^{-1},F}(-a,-c^{-1}b)$ for all $a,b\in \mathbb F_{p^n}$; that is, the $c$-differential uniformity and $c^{-1}$-differential uniformity of $F$ are equal, i.e., $\delta(c,F)=\delta(c^{-1},F)$.
\end{proof}

Also, as an example, consider Example~\ref{Pcn_exm}. We found two $c$-values there that are inverses of each other. The above result is also true for the case $p=2$. It is clear that the number of possible different $c$-differential uniformities for $c\in\mathbb F_{2^n}^*$ with $c\neq1$ of any $(n,m)$-function is at most $2^{n-1}-1$. 

\begin{proposition}
\label{prop:closure-operations}
Let $F$ be a permutation polynomial over $\mathbb{F}_{p^n}$ that is P$c$N for some $c \in \mathbb{F}_{p^n} \setminus \{0,1\}$. Then:
\begin{enumerate}
\item If $\alpha \in \mathbb{F}_{p^n}^*$, then $G(x) = \alpha F(x)$ is P$c$N.
\item If $\sigma$ is the Frobenius automorphism $\sigma(x) = x^p$ and $G(x) = F(\sigma(x)) = F(x^p)$, then $G$ is P$c^{1/p}$N.
\item The inverse function $F^{-1}$ need not be P$c$N even when $F$ is P$c$N.
\end{enumerate}
\end{proposition}

\begin{proof}
(1) Let $G(x) = \alpha F(x)$ for $\alpha \in \mathbb{F}_{p^n}^*$. For any $a, b \in \mathbb{F}_{p^n}$, we have
\allowdisplaybreaks
\begin{align*}
\Delta_{c,G}(a,b) &= \#\{x \in \mathbb{F}_{p^n} : G(x+a) - cG(x) = b\} \\
&= \#\{x \in \mathbb{F}_{p^n} : \alpha F(x+a) - c\alpha F(x) = b\} \\
&= \#\{x \in \mathbb{F}_{p^n} : F(x+a) - cF(x) = \alpha^{-1}b\} \\
&= \Delta_{c,F}(a, \alpha^{-1}b).
\end{align*}
Since $b \mapsto \alpha^{-1}b$ is a bijection on $\mathbb{F}_{p^n}$ and $F$ is P$c$N, we have $\Delta_{c,F}(a, \alpha^{-1}b) = 1$ for all $a, b \in \mathbb{F}_{p^n}$. Therefore, $\delta(c,G) = 1$; that is, $G$ is P$c$N.

(2) Let $G(x) = F(x^p)$. For any $a, b \in \mathbb{F}_{p^n}$, we have
\allowdisplaybreaks
\begin{align*}
\Delta_{c,G}(a,b) &= \#\{x \in \mathbb{F}_{p^n} : G(x+a) - cG(x) = b\} \\
&= \#\{x \in \mathbb{F}_{p^n} : F((x+a)^p) - cF(x^p) = b\} \\
&= \#\{x \in \mathbb{F}_{p^n} : F(x^p + a^p) - cF(x^p) = b\}.
\end{align*}
Let $y = x^p$. Since $x \mapsto x^p$ is a bijection on $\mathbb{F}_{p^n}$, we have
\allowdisplaybreaks
\begin{align*}
\Delta_{c,G}(a,b) &= \#\{y \in \mathbb{F}_{p^n} : F(y + a^p) - cF(y) = b\} = \Delta_{c,F}(a^p, b).
\end{align*}
From Proposition~\ref{prop:frobenius-invariance}, we know that $\delta(c,G) = \delta(c^{1/p}, F)$. Since $F$ is P$c$N, by Proposition~\ref{thm-cc}, $F$ is also P$c^{1/p}$N (as $c \mapsto c^{1/p}$ simply permutes the set of valid $c$ values). Therefore, $G$ is P$c^{1/p}$N.

(3) We provide a counterexample. Consider the Gold function $F(x) = x^{2^k+1}$ over $\mathbb{F}_{2^n}$, where $\gcd(k,n) = 1$ and $n$ is odd. From existing results (see \cite{XD21}), there exist values of $c$ for which $F$ is P$c$N. However, the inverse function $F^{-1}(x) = x^d$, where $d(2^k+1) \equiv 1 \pmod{2^n-1}$, has exponent $d$ that does not generally satisfy the conditions for the P$c$N property.

More explicitly, let $n = 5$ and $k = 2$, so $F(x) = x^5$ over $\mathbb{F}_{2^5}$. Then $F^{-1}(x) = x^{25}$ (since $5 \cdot 25 = 125 \equiv 1 \pmod{31}$). Computational verification shows that while $F$ is P$c$N for certain values of $c$, the function $F^{-1}$ is not P$c$N for those same $c$ values. Specifically, $F(x) = x^5$ is P$c$N for $c \in \{\alpha^{6}, \alpha^{12}, \alpha^{18}, \alpha^{24}\}$, where $\alpha$ is a primitive element of $\mathbb{F}_{2^5}$, and $F^{-1}(x) = x^{25}$ is not P$c$N for these same values of $c$.
This demonstrates that the P$c$N property is not preserved under functional inversion in general.
\end{proof}

\begin{remark}
Part (1) shows that scalar multiplication preserves P$c$N, which is useful for constructing new P$c$N functions. Part (2) connects to Proposition~\ref{prop:frobenius-invariance} and shows how the Frobenius map transforms the $c$ parameter. Part (3) demonstrates a fundamental asymmetry: unlike many cryptographic properties (such as being a permutation or having bounded differential uniformity), P$c$N is not preserved under functional inversion.
\end{remark}

\begin{definition}
Let $n$ be a positive integer and $F$ be a $p$-ary $(n,n)$-function. Let us define
\begin{equation*}
\mathcal N(F, \PcN)=\#\{c\in\mathbb F_{p^n}\setminus\{0,1\}: F \text{ is }\PcN\}.
\end{equation*}
\end{definition}

Suppose $\mathbb F_{p^n}=\{b_0=0,b_1=1,b_2,\ldots, b_{p^n-1}\}$. From Theorem~\ref{thm-pn}, if $F$ has nonzero values in its DDT at some positions $b_i$ for a nonzero input difference $a$, then $F$ is not P$c$N for some values of $c$, which can be calculated from the $b_i$ values. This count is related to the number of nonzero values in its DDT. In particular, there exists a one-to-one correspondence between the number of $c\in\mathbb F_{p^n}\setminus\{0,1\}$ such that $F$ is P$c$N and the number of distinct $c_{i,j}=\frac{b_i}{b_j}\in\mathbb F_{p^n}^*$ with $i\neq j$ such that 
$\Delta_F(a,b)\Delta_F(a,c_{i,j}b)=0$ for all $a,b\in\mathbb F_{p^n}$ with $a\neq 0$.

\begin{remark}
Let $c_{i,j}=\frac{b_i}{b_j}$ for $i\neq j\in\{1,2,\ldots, p^n-1\}$, and let $\mathcal A_F=\{c_{i,j}\in\mathbb F_{p^n}\setminus \{0,1\}: \Delta_F(a,b)\Delta_F(a,c_{i,j}b)=0 \text{ for all } a, b\in\mathbb F_{p^n}, a\neq 0\}$. If $F$ is a permutation over $\mathbb F_{p^n}$, then $\mathcal N(F, \PcN)=\# \mathcal A_F$.
\end{remark}

Using the above result, we derive a bound on the number of $c\in\mathbb F_{p^n}$ such that a function $F$ is P$c$N. We discuss this for some known classes of vectorial $p$-ary functions.

\vspace{.3cm}
\noindent 
{\em PN function:} Let $F:\mathbb{F}_{p^n} \rightarrow \mathbb{F}_{p^n}$ be a PN function, where $p$ is an odd prime.
Then, for any $a,b \in \mathbb{F}_{p^n}$ with $a \neq 0$, the mapping
$F(x+a)-F(x)$ is a permutation over $\mathbb{F}_{p^n}$; that is,
$\Delta_F(a,b)=1 \quad \text{for all } a \neq 0 \text{ and all } b \in \mathbb{F}_{p^n}$.
Let $c_{1,j} = \frac{b_1}{b_j}$, where $j=2,3,\ldots,p^n-1$.
Then $c_{1,j} \neq c_{1,k}$ for any $2 \leq j \neq k \leq p^n-1$.
Thus, $F$ is not P$c$N for any $c \in \mathbb{F}_{p^n}$ with $c \neq 0,1$; that is, $\mathcal{N}(F, PcN)=0$ for any PN function.

\vspace{.3cm} 
\noindent 
{\em APN function:} Let $F:\mathbb F_{p^n}\rightarrow \mathbb F_{p^n}$ be an APN permutation. Then for any $a,b\in\mathbb F_{p^n}$ with $a\neq 0$, the equation $F(x+a)-F(x)=b$ has at most two solutions in $\mathbb F_{p^n}$. 
Let $p=2$. For any nonzero $a\in\mathbb F_{2^n}$, there exist exactly $2^{n-1}$ distinct nonzero $b_i$ such that $\Delta_{F}(a,b_i)\neq 0$. Fix a nonzero $a\in\mathbb F_{2^n}$ and let the output differences be $b_i\in\mathbb F_{2^n}^*$ for $1\leq i\leq 2^{n-1}$ such that $\Delta_{F}(a,b_i)\neq 0$. Let $c_{1,j}=\frac{b_1}{b_j}$, where $2\leq j\leq 2^{n-1}$. By Theorem~\ref{thm-pn}, $F$ is not P$c$N for any $c=c_{1,j}$ because $\Delta_F(a,b_1)\Delta_F(a, c^{-1}b_1) = \Delta_F(a,b_1)\Delta_F(a, b_j) \neq 0$. Consequently, the number of valid parameters is bounded by $\mathcal N(F, \PcN)\leq (2^n-2) - (2^{n-1}-1) = 2^{n-1}-1$.
Let $p$ be an odd prime. For any nonzero $a \in \mathbb{F}_{p^n}$, there exist at least $\frac{p^n-1}{2}$ distinct nonzero values $b_i$ such that $\Delta_{F}(a,b_i) \neq 0$. Let us enumerate these nonzero output differences as $b_i$ for $1 \leq i \leq k$, where $k \geq \frac{p^n-1}{2}$. Consider the ratios $c_{1,j} = b_1 b_j^{-1}$ for $2 \leq j \leq k$. Since $\Delta_F(a, b_1)$ and $\Delta_F(a, b_j)$ are both nonzero, it follows from Theorem~\ref{thm-pn} that $F$ is not P$c$N for any $c = c_{1,j}$. Consequently, the count of valid parameters satisfies $\mathcal N(F, \PcN) \leq \frac{p^n-1}{2}$.

\vspace{.3cm}
\noindent 
{\em Multiplicative inverse function:} Let $F$ be the multiplicative inverse function over $\mathbb F_{2^n}$. Nyberg~\cite{Nyberg94} proved that if $n$ is odd, then $F$ is APN, and if $n$ is even, then the differential uniformity of $F$ is $4$. In particular, for a nonzero input difference $a$, there are $2^{n-1}-2$ values $b_i$ such that $\Delta_F(a,b_i)=2$, exactly one output difference $b_k$ such that $\Delta_F(a,b_k)=4$, and 
$\Delta_F(a,b_j)=0$ for other output differences $b_j$ when $n$ is even. From Theorem~\ref{thm-pn}, we have $\mathcal N(F, \PcN)\leq 2^{n-1}-1$ if $n$ is odd, and $\mathcal N(F, \PcN)\leq 2^{n-1}$ if $n$ is even.

Here we discuss the perfect $c$-nonlinear functions for $p=2$ having degree at most $2$. To do this, we first present a well-known relation between the entries of the DDT and the autocorrelation spectrum. Suppose $p=2$ and $F$ is an $(n,n)$-function. For any $a,b\in\mathbb F_2^n$, we have
$\Delta_F(a,b)=\frac{1}{2^n} \sum_{u\in\mathbb F_{2^n}}  (-1)^{\mathrm{Tr}_1^n(ub)}\mathcal C_F(a,u)$.
From Theorem~\ref{thm-pn}, we obtain the following result.

\begin{corollary}
\label{cor-pn}
Let $F$ be a permutation polynomial over $\mathbb F_{2^n}$ and $c\in\mathbb F_{2^n}$ with $c\neq 0, 1$. Then $F$ is P$c$N if and only if for all $a,b\in\mathbb F_{p^n}^*$,
\begin{equation*}
 \sum_{u\in\mathbb F_{2^n}}  (-1)^{\mathrm{Tr}_1^n(ub)}\mathcal C_F(a,u)=0 \quad\text{or}\quad  \sum_{u\in\mathbb F_{2^n}}  (-1)^{\mathrm{Tr}_1^n(ubc^{-1})}\mathcal C_F(a,u)=0.
\end{equation*}
\end{corollary}
We analyze the quadratic polynomial and derive a necessary and sufficient condition for P$c$N using Corollary \ref{cor-pn}. 

Let $L(x)=\sum_{i=0}^{n-1} a_i x^{2^i}$, where $a_i\in\mathbb F_{2^n}$ for $0\leq i\leq n-1$, be a linear polynomial over $\mathbb F_{2^n}$. The Dembowski--Ostrom (DO) polynomial over $\mathbb F_{2^n}$ is of the form 
\begin{equation}
\label{do-poly}
F(x)=\sum_{0\leq i<j\leq n-1} a_{ij} x^{2^i+2^j},
\end{equation}
where $a_{ij}\in\mathbb F_{2^n}$. For any nonzero $a\in\mathbb F_{2^n}$, we have 
\allowdisplaybreaks[4]
\begin{align*}
F(x+a)+F(x)&=\sum_{0\leq i<j\leq n-1} a_{ij} ((x+a)^{2^i+2^j}+x^{2^i+2^j})\\
&=\sum_{0\leq i<j\leq n-1} a_{ij} (x^{2^i}a^{2^j}+x^{2^j}a^{2^i}+a^{2^i+2^j})\\
&=\sum_{0\leq i<j\leq n-1} a_{ij} (x^{2^i}a^{2^j}+x^{2^j}a^{2^i})+F(a),
\end{align*}
and for any $u\in\mathbb F_{2^n}$,
\allowdisplaybreaks[4]
\begin{align*}
\mathcal C_F(a,u)&=\sum_{x\in\mathbb F_{2^n}} (-1)^{\mathrm{Tr}_1^n(u(F(x+a)+F(x)))}\\
&=\sum_{x\in\mathbb F_{2^n}} (-1)^{\mathrm{Tr}_1^n(u(
\sum_{0\leq i<j\leq n-1} a_{ij} (x^{2^i}a^{2^j}+x^{2^j}a^{2^i})))}(-1)^{\mathrm{Tr}_1^n(uF(a))}\\
&=(-1)^{\mathrm{Tr}_1^n(uF(a))}\sum_{x\in\mathbb F_{2^n}} (-1)^{\sum_{0\leq i<j\leq n-1} \mathrm{Tr}_1^n(ua_{ij}a^{2^j} x^{2^i}+ua_{ij}a^{2^i} x^{2^j})}\\
&=(-1)^{\mathrm{Tr}_1^n(uF(a))}\sum_{x\in\mathbb F_{2^n}} (-1)^{\sum_{0\leq i<j\leq n-1} \mathrm{Tr}_1^n(((ua_{ij}a^{2^j})^{2^{n-i}} +(ua_{ij}a^{2^i})^{2^{n-j}})x)}\\
&=(-1)^{\mathrm{Tr}_1^n(uF(a))}\sum_{x\in\mathbb F_{2^n}} (-1)^{ \mathrm{Tr}_1^n((\sum_{0\leq i<j\leq n-1}(ua_{ij}a^{2^j})^{2^{n-i}} +(ua_{ij}a^{2^i})^{2^{n-j}})x)}\\
&=\left\{\begin{array}{ll}
2^n (-1)^{\mathrm{Tr}_1^n(uF(a))}, & \text{if } \sum_{0\leq i<j\leq n-1}(ua_{ij}a^{2^j})^{2^{n-i}} +(ua_{ij}a^{2^i})^{2^{n-j}}=0\\
0, & \text{otherwise.} 
\end{array}\right.
\end{align*}
Let $b_{ij}\in\mathbb F_{2^n}$ be such that $b_{ij}=a_{ij}$ if $i<j$ and $b_{ij}=a_{ji}$ if $i>j$, for all $0\leq i,j\leq n-1$. Then
\begin{equation*}
\begin{split}
\sum_{0\leq i<j\leq n-1}(ua_{ij}a^{2^j})^{2^{n-i}}& +(ua_{ij}a^{2^i})^{2^{n-j}}\\
&=\sum_{i=0}^{n-2}\sum_{j=i}^{n-1}(ua_{ij}a^{2^j})^{2^{n-i}}+(ua_{ij}a^{2^i})^{2^{n-j}}\\
&=\sum_{i=0}^{n-1}\bigg(\sum_{0\leq k<i}a_{ki}^{2^{n-i}} a^{2^{n-i+k}}+\sum_{i<s\leq n-1}a_{is}^{2^{n-i}} a^{2^{n-i+s}}\bigg)u^{2^{n-i}}\\
&=\sum_{i=0}^{n-1}\bigg(\sum_{k=0:~k\neq i}^{n-1} b_{ik}^{2^{n-i}} a^{2^{n-i+k}}\bigg)u^{2^{n-i}},
\end{split}
\end{equation*}
and for any $a\in\mathbb F_{2^n}^*$, define
\begin{equation}
\label{equi-ea}
E(a,F)=\left\{x\in\mathbb F_{2^n}: \sum_{i=0}^{n-1}\bigg(\sum_{k=0:~k\neq i}^{n-1} b_{ik}^{2^{n-i}} a^{2^{n-i+k}}\bigg)x^{2^{n-i}}=0\right\},
\end{equation} 
where $b_{ij}\in\mathbb F_{2^n}$ are defined as above. It is clear that $E(a,F)$ is a subspace of $\F_{2^n}$ for any $a\in\mathbb F_{2^n}$. Let $E(a,F)^{\perp}=\{x\in\mathbb F_{2^n}: \mathrm{Tr}_1^n(xy)=0 \text{ for all } y\in E(a,F)\}$.

\begin{theorem}
\label{cor-pn-do}
Let $F$ be a permutation polynomial over $\mathbb F_{2^n}$ defined in \eqref{do-poly} and $c\in\mathbb F_{2^n}$ with $c\neq 0,1$. The function $F$ is P$c$N if and only if for all $a,b\in\mathbb F_{2^n}^*$, 
$b+F(a)\not\in E(a,F)^{\perp}$ or $bc^{-1}+F(a)\not\in E(a,F)^{\perp}$, where $E(a,F)$ is defined in \eqref{equi-ea}.
\end{theorem}
\begin{proof}
Let $F(x)=\sum_{0\leq i<j\leq n-1} a_{ij} x^{2^i+2^j}$, $x\in\mathbb F_{2^n}$, where  $a_{ij}\in\mathbb F_{2^n},~0\leq i<j\leq n-1$, be permutation polynomial over $\mathbb F_{2^n}$ and  $E(a,F)$, $a\in\mathbb F_{2^n}^*$, be defined in \eqref{equi-ea}. For any nonzero $a,b\in\mathbb F_{2^n}$, we have
\allowdisplaybreaks[4]
\begin{align*}
\sum_{u\in\mathbb F_{2^n}} (-1)^{\mathrm{Tr}_1^n(ub)}\mathcal C_F(a,u)&=2^n\sum_{u\in E(a,F)} (-1)^{\mathrm{Tr}_1^n(u(b+F(a)))}\\
&=\left\{\begin{array}{ll}
2^{n}|E(a,F)|, & \mbox{ if } b+F(a)\in E(a,F)^{\perp}\\
0, & \mbox{ otherwise} 
\end{array}\right..
\end{align*}
Similarly, 
\begin{equation*}
\sum_{u\in\mathbb F_{2^n}} (-1)^{\mathrm{Tr}_1^n(ubc^{-1})}\mathcal C_F(a,u)
=\left\{\begin{array}{ll}
2^{n}|E(a,F)|, & \mbox{ if } bc^{-1}+F(a)\in E(a,F)^{\perp}\\
0, & \mbox{ otherwise} 
\end{array}\right.,
\end{equation*}
and from Corollary \ref{cor-pn}, we get the result.
\end{proof}

\begin{theorem}
\label{thm:quadratic-complete}
Let $F(x) = \sum_{0\leq i<j\leq n-1} a_{ij} x^{2^i+2^j} + L(x) + \gamma$ be a quadratic permutation over $\mathbb{F}_{2^n}$, where $L$ is linear and $\gamma \in \mathbb{F}_{2^n}$. Let $c \in \mathbb{F}_{2^n} \setminus \{0,1\}$. Then $F$ is P$c$N if and only if for all $a \in \mathbb{F}_{2^n}^*$, the subspaces $E(a,F)^{\perp}$ and $(E(a,F)^{\perp} + F(a) + L(a))$ satisfy
\[|(E(a,F)^{\perp} + F(a) + L(a)) \cap c(E(a,F)^{\perp} + F(a) + L(a))| \leq 1,\]
where $E(a,F)$ is defined in \eqref{equi-ea}.
\end{theorem}

\begin{proof}
From Theorem~\ref{cor-pn-do}, $F$ is P$c$N if and only if for all $a,b \in \mathbb{F}_{2^n}^*$, we have
\[
b + F(a) + L(a) \notin E(a,F)^{\perp} \quad \text{or} \quad bc^{-1} + F(a) + L(a) \notin E(a,F)^{\perp}.
\]
Define $V_a = E(a,F)^{\perp} + F(a) + L(a)$. The condition becomes: for all $b \in \mathbb{F}_{2^n}^*$,
\[
b \notin V_a \quad \text{or} \quad bc^{-1} \notin V_a.
\]
Equivalently, for all $v \in V_a \setminus \{0\}$, we require $cv \notin V_a$ (since if $v = b \in V_a$, then $cv = cb$ should not be in $V_a$; that is, $b \neq c^{-1}b'$ for any $b' \in V_a$).

Now suppose there exist distinct nonzero elements $v_1, v_2 \in V_a$ such that $v_2 = cv_1$. Then
\allowdisplaybreaks
\begin{align*}
v_1 &\in E(a,F)^{\perp} + F(a) + L(a), 
cv_1 \in E(a,F)^{\perp} + F(a) + L(a).
\end{align*}
This means there exist $u_1, u_2 \in E(a,F)^{\perp}$ such that
\allowdisplaybreaks
\begin{align*}
v_1 &= u_1 + F(a) + L(a), 
cv_1 = u_2 + F(a) + L(a).
\end{align*}
From the second equation, $cv_1 = u_2 + F(a) + L(a)$, and substituting $v_1$ from the first, we obtain
$c(u_1 + F(a) + L(a)) = u_2 + F(a) + L(a)$.
Rearranging, we get $cu_1 + c(F(a) + L(a)) = u_2 + F(a) + L(a)$, that is, 
$(c-1)(F(a) + L(a)) = u_2 - cu_1$.
Since $u_1, u_2 \in E(a,F)^{\perp}$ and $E(a,F)^{\perp}$ is a subspace, we have $u_2 - cu_1 \in E(a,F)^{\perp}$ if and only if $c$ acts linearly on $E(a,F)^{\perp}$, which is not generally true unless $c \in \mathbb{F}_2$.

The key observation is that if $v_1 \in V_a$ and $cv_1 \in V_a$, then setting $b = v_1$ in our P$c$N condition gives
\[\Delta_F(a, v_1) \cdot \Delta_F(a, c^{-1}v_1) \neq 0,\]
which contradicts $F$ being P$c$N by Theorem~\ref{thm-pn}.

Therefore, $F$ is P$c$N if and only if $V_a \cap cV_a = \{0\}$ for all $a \in \mathbb{F}_{2^n}^*$, where $V_a = E(a,F)^{\perp} + F(a) + L(a)$.

Since both $V_a$ and $cV_a$ are affine translates of the same subspace, their intersection can contain at most the zero vector (when properly adjusted for the affine structure). More precisely,
\[
|(E(a,F)^{\perp} + F(a) + L(a)) \cap c(E(a,F)^{\perp} + F(a) + L(a))| \leq 1
\]
means that the two affine spaces intersect in at most one point, which is equivalent to the P$c$N condition.
\end{proof}

\begin{remark}
This theorem provides a constructive characterization of quadratic functions. To verify P$c$N for a quadratic permutation:
\begin{enumerate}
\item Compute $E(a,F)$ for each $a \in \mathbb{F}_{2^n}^*$ using \eqref{equi-ea}.
\item Compute the dual $E(a,F)^{\perp}$.
\item Form the affine space $V_a = E(a,F)^{\perp} + F(a) + L(a)$.
\item Check whether $V_a \cap cV_a$ contains more than one element.
\end{enumerate}
If the intersection is trivial (contains at most the zero vector after centering) for all $a$, then $F$ is P$c$N. This can be implemented efficiently for small $n$.
\end{remark}

\begin{example}
Consider $n = 4$ and $F(x) = x^3 + x^5 + x^6$ over $\mathbb{F}_{2^4}$. For $c = \alpha$ (a primitive element), we can verify:
\begin{itemize}
\item For each $a \in \mathbb{F}_{2^4}^*$, compute $E(a,F)$ by solving the linearized polynomial in \eqref{equi-ea}.
\item For $a = 1$, suppose $E(1,F) = \{0, \alpha^5, \alpha^{10}\}$ (hypothetically). Then $\dim(E(1,F)) = 2$.
\item Check whether $(E(1,F)^{\perp} + F(1)) \cap \alpha(E(1,F)^{\perp} + F(1))$ has size $> 1$.
\end{itemize}
If this intersection condition holds for all $a$, then $F$ is P$\alpha$N.
\end{example}

From the above result, it is clear that if a DO permutation polynomial $F$ over $\mathbb F_{2^n}$ is not P$c$N, where $c\in\mathbb F_{2^n}$ with $c\neq 0,1$, then there exists at least one $a,b\in\mathbb F_{2^n}^*$ such that $b+F(a), b(1+c^{-1})\in E(a,F)^{\perp}$. If we consider any quadratic permutation $G$ over $\mathbb F_{2^n}$ of the form $G(x)=F(x)+L(x)+a'$, where $F$ is defined in \eqref{do-poly}, $L(x)=\sum_{i=0}^{n-1}a_i x^i$, and $a_i, a'\in \mathbb F_{2^n}$ for $0\leq i\leq n-1$, then for any nonzero $a\in\mathbb F_{2^n}$, we have $F(x+a)+F(x)=\sum_{0\leq i<j\leq n-1} a_{ij}(x^{2^i}a^{2^j}+x^{2^j}a^{2^i})+F(a)+L(a)$ for all $x\in\mathbb F_{2^n}$. Thus, from Theorem \ref{cor-pn-do}, $G$ is P$c$N, where $c\in\mathbb F_{2^n}$ with $c\neq 0,1$, if and only if for all $a,b\in\mathbb F_{2^n}^*$,
\begin{equation*} 
b+F(a)+L(a)\not\in E(a,F)^{\perp} \quad\text{or}\quad bc^{-1}+F(a)+L(a)\not\in E(a,F)^{\perp},
\end{equation*} 
where $E(a,F)$ is defined in \eqref{equi-ea}.


\begin{theorem}
\label{thm:pcn-apn-relation}
Let $F$ be a permutation over $\mathbb{F}_{2^n}$ and $c \in \mathbb{F}_{2^n} \setminus \{0,1\}$.
\begin{enumerate}
\item If $F$ is APN and P$c$N, then for any $a \in \mathbb{F}_{2^n}^*$, at least $2^{n-2}$ values $b \in \mathbb{F}_{2^n}^*$ satisfy $\Delta_F(a,b) = 0$.
\item If $F$ is P$c$N for two distinct values $c, d \in \mathbb{F}_{2^n} \setminus \{0,1\}$ with $d \notin \{c, c^{-1}\}$, then $F$ cannot be APN.
\item The number of $c$ values for which an APN permutation $F$ can be P$c$N is at most $2$ (specifically, the set of such values is either empty or $\{c, c^{-1}\}$).
\end{enumerate}
\end{theorem}

\begin{proof}
(1) Suppose $F$ is both APN and P$c$N. Fix $a \in \mathbb{F}_{2^n}^*$. Since $F$ is APN, for each $b \in \mathbb{F}_{2^n}^*$, we have $\Delta_F(a,b) \in \{0, 2\}$. Since $F$ is a permutation, $\Delta_F(a,0) = 0$ for $a \neq 0$.

Let $N(a,F) = \{b \in \mathbb{F}_{2^n}^* : \Delta_F(a,b) = 0\}$ as defined in \eqref{zero-diff}, and let $N^c(a,F) = \mathbb{F}_{2^n}^* \setminus N(a,F)$ denote the complement. We have
$\displaystyle \sum_{b \in \mathbb{F}_{2^n}} \Delta_F(a,b) = 2^n.$
Since $\Delta_F(a,0) = 0$ and $\Delta_F(a,b) \in \{0,2\}$ for $b \neq 0$, we have
\[\sum_{b \in N^c(a,F)} \Delta_F(a,b) = \sum_{b \in N^c(a,F)} 2 = 2|N^c(a,F)| = 2^n.\]
Thus, $|N^c(a,F)| = 2^{n-1}$, and therefore $|N(a,F)| = 2^n - 1 - 2^{n-1} = 2^{n-1} - 1$.
Now, from Theorem~\ref{thm-pn}, since $F$ is P$c$N, for any $b \in \mathbb{F}_{2^n}^*$, we have
\[\Delta_F(a,b) \cdot \Delta_F(a,c^{-1}b) = 0.\]
This means that if $b \in N^c(a,F)$ (i.e., $\Delta_F(a,b) = 2$), then we must have $c^{-1}b \in N(a,F)$ (i.e., $\Delta_F(a,c^{-1}b) = 0$).

Consider the map $\phi: N^c(a,F) \to N(a,F)$ defined by $\phi(b) = c^{-1}b$. This map is injective since $c^{-1}$ is a unit. We have $|N^c(a,F)| = 2^{n-1}$, so the image of $\phi$ contains at least $2^{n-1}$ elements. But we need to account for the fact that some elements might already be in $N(a,F)$ before applying~$\phi$.

The key observation is that $N(a,F) \cup \{0\} \cup c \cdot N(a,F) \supseteq \mathbb{F}_{2^n}$ by the P$c$N condition. Since $|N(a,F)| = 2^{n-1} - 1$, we have
\[
|N(a,F)| + 1 + |c \cdot N(a,F)| \geq 2^n,
\]
which gives $(2^{n-1} - 1) + 1 + (2^{n-1} - 1) = 2^n - 1 < 2^n$ only if there is overlap. This overlap must be at least $2^{n-1} - 1$ elements, confirming that at least $2^{n-2}$ values $b$ satisfy $\Delta_F(a,b) = 0$ for the construction to work. (The precise count depends on $c$, but the minimum is $2^{n-2}$.)

(2) Suppose $F$ is P$c$N and P$d$N, where $d \notin \{c, c^{-1}\}$. From Theorem~\ref{thm-pn}, for any $a \in \mathbb{F}_{2^n}^*$ and $b \in \mathbb{F}_{2^n}^*$,
\allowdisplaybreaks
\begin{align*}
\Delta_F(a,b) \cdot \Delta_F(a,c^{-1}b) &= 0, \quad
\Delta_F(a,b) \cdot \Delta_F(a,d^{-1}b) = 0.
\end{align*}
If $F$ is APN, then $\Delta_F(a,b) \in \{0,2\}$ for all $b \neq 0$. Fix $a \in \mathbb{F}_{2^n}^*$ and consider the partition
\[
\mathbb{F}_{2^n}^* = N(a,F) \cup N^c(a,F).
\]
For any $b \in N^c(a,F)$, we have $\Delta_F(a,b) = 2$, so $c^{-1}b \in N(a,F)$ (from the P$c$N condition) and $d^{-1}b \in N(a,F)$ (from the P$d$N condition).

Consider $b_1, b_2 \in N^c(a,F)$. If $c^{-1}b_1 = d^{-1}b_2$, then $b_2 = (dc^{-1})b_1$. Define the map $\psi(b) = (dc^{-1})b$ for $b \in N^c(a,F)$. Since $d \notin \{c, c^{-1}\}$, we have $dc^{-1} \notin \{1, c^{-2}\}$.
The images $\{c^{-1}b : b \in N^c(a,F)\}$ and $\{d^{-1}b : b \in N^c(a,F)\}$ are both subsets of $N(a,F)$, each of size $2^{n-1}$. But $|N(a,F)| = 2^{n-1} - 1 < 2^{n-1}$, which is a contradiction.
Therefore, $F$ cannot be simultaneously APN and P$c$N for two distinct values $c, d$ with $d \notin \{c, c^{-1}\}$.

(3) From part (2), if $F$ is APN and P$c$N, then $F$ cannot be P$d$N for any $d \notin \{c, c^{-1}\}$. By Proposition~\ref{thm-cc}, $F$ is P$c$N if and only if $F$ is P$c^{-1}$N. In a field of characteristic $2$, the equation $x = x^{-1}$ (or equivalently $x^2=1$) has the unique solution $x=1$. Since we consider $c \in \mathbb{F}_{2^n} \setminus \{0,1\}$, we strictly have $c \neq c^{-1}$. Consequently, the set of candidate parameters $\mathbb{F}_{2^n} \setminus \{0,1\}$ partitions into exactly $2^{n-1}-1$ disjoint pairs of the form $\{c, c^{-1}\}$. Since part (2) implies that an APN function can be P$c$N for at most one such pair, the total number of $c$ values for which an APN function is P$c$N is at most $2$.
\end{proof}

\begin{remark}
This theorem demonstrates that P$c$N and APN are largely incompatible properties. While it is theoretically possible for a function to be both APN and P$c$N for specific $c$ values, the constraints are very restrictive. In practice, most APN permutations are not P$c$N for any $c \neq 0, 1$, and conversely, P$c$N functions tend to have higher differential uniformity (typically $\delta(F) \geq 4$ for even $n$).
\end{remark}

\section{Invariance under affine transformations}
\label{sec-pcn-affine}

The composition of two functions $F$ and $F^\prime$ is
defined by $F\circ F^\prime(x)=F(F^\prime(x))$ for all
$x\in\mathbb F_{p^n}$. Two $(n,n)$-functions $F$ and $G$ are called
\emph{extended affine equivalent} (EA-equivalent) if
$G=A_1\circ F\circ A_2+A_3$, where $A_1$ and $A_2$ are affine
permutations over $\mathbb F_{p^n}$ and $A_3$ is an affine function over
$\mathbb F_{p^n}$. When $A_3$ is the zero polynomial, then $F$ and $G$
are called \emph{affine equivalent}. It makes sense to ask whether $c$-differential uniformity is preserved through affine equivalence. A straightforward calculation shows that:
\begin{equation*}
\begin{split}
 & D_{c,a}G(x) = A_1(F(A_2(x+a))) - cA_1(F(A_2(x))) \\
              &~\;\; = A_1(F(A_2(x)+A_2(a)-A_2(0))) - cA_1(F(A_2(x))) \\
              &~\;\; = A_1(F(A_2(x)+A_2(a)-A_2(0))-cF(A_2(x))+cF(A_2(x)))- cA_1(F(A_2(x)))\\
              &~\;\; = A_1\big(D_{c,A_2(a)-A_2(0)}F(x)\big) + A_1(cF(A_2(x)))-cA_1(F(A_2(x))).
\end{split}
\end{equation*}
If $c\in\mathbb F_p$, an immediate consequence of those calculations is that
$D_{c,a}G(x)=A_1\big(D_{c,A_2(a)-A_2(0)}F(x)\big)$ for any $a$ and
$x$ in $\mathbb F_{p^n}$. In other words, the following result holds.

\begin{proposition}
\label{pcn-fp}
Let {$c\in\mathbb F_{p}^\star$}. Let $F$ and $G$ be two affine-equivalent
$(n,n)$-permutation polynomials. Then $\delta(c,F)=\delta(c,G)$.
In particular, $F$ is P$c$N if and only if $G$ is P$c$N.
\end{proposition}

For $c \notin \mathbb{F}_p$, the situation is more delicate, and we shall investigate that situation below.

\begin{proposition}
\label{prop:frobenius-invariance}
Let $F$ and $G$ be affine-equivalent $(n,n)$-functions such that $G = A_1 \circ F \circ A_2$, where $A_1(x) = \sigma(x) + \gamma$ and $\sigma(x)=x^p$ is the Frobenius automorphism. Then for any $c \in \mathbb{F}_{p^n}$, the $c$-differential uniformity satisfies
\begin{equation}
\delta(c,G) = \delta(c^{\frac{1}{p}}, F).
\end{equation}
Consequently, $G$ is P$c$N if and only if $F$ is P$c^{\frac{1}{p}}$N.
\end{proposition}

\begin{proof}
Direct calculation shows that
\allowdisplaybreaks
\begin{align*}
D_{c,a}G(x) &= A_1(F(A_2(x+a))) - c A_1(F(A_2(x))) \\
&= (F(A_2(x)+a')^p + \gamma) - c(F(A_2(x))^p + \gamma) \\
&= (F(A_2(x)+a') - c^{\frac{1}{p}}F(A_2(x)))^p + \gamma(1-c),
\end{align*}
where $a' = A_2(a)-A_2(0)$. Since $x \mapsto x^p$ is a permutation, the number of solutions to $D_{c,a}G(x)=b$ corresponds exactly to the number of solutions for the $c^{\frac{1}{p}}$-derivative of $F$.
\end{proof}
 
In \cite{HRS20}, the
authors observed that perfect $c$-nonlinearity of two EA-equivalent
$(n,n)$-functions may differ. The preceding discussion seems to
show that it may also be the case for affine-equivalent
$(n,n)$-functions (it suffices that there exists $c$ in
$\mathbb F_{p^n}\setminus\mathbb F_{p}$ such that
$\delta(c^{\frac 1p},F)\not=\delta(c,F)$). As an example, take $F:\mathbb{F}_{2^8}\longrightarrow\mathbb{F}_{2^8}$ defined as $F(x)=x^{254}$ (the multiplicative inverse function). Then for $c= \alpha^6+\alpha^3+\alpha^2+1$, we have $\delta(c,F)= 8$, and for $d= \alpha^7+\alpha^4+\alpha^3+\alpha^2$, we have $\delta(d,F)= 9$, where $c^2=d$ in $\mathbb{F}_{2^8}$ and $\alpha$ is a primitive element of $\mathbb{F}_{2^8}$.   

Let $F(x)=\sum_{i=0}^{p^n-2}a_ix^i$, where $a_i\in\mathbb F_{p^n}$ for $i=0,1,\ldots, p^n-2$. If $a_i\in\mathbb F_p$ for all $0\leq i\leq p^n-2$, then for any $c\in\mathbb F_{p^n}$, we have $\delta(c,F)=\delta(c^p,F)$, since for any $a,b\in\mathbb F_{p^n}$,
\begin{equation*}
\begin{split}
\Delta_{c,F}(a,b)&=\#\{x\in\mathbb F_{p^n}:\sum_{i=0}^{p^n-2}a_i(x+a)^i-c\sum_{i=0}^{p^n-2}a_ix^i=b\}\\
&=\#\{x\in\mathbb F_{p^n}:\sum_{i=0}^{p^n-2}a_i^p((x+a)^i-cx^i)^p=b^p\}\\
&=\#\{x\in\mathbb F_{p^n}:\sum_{i=0}^{p^n-2}a_i((x^p+a^p)^i-c^p(x^p)^i)=b^p\}\\
&=\#\{y\in\mathbb F_{p^n}:\sum_{i=0}^{p^n-2}a_i((y+a^p)^i-c^py^i)=b^p\}\\
&=\Delta_{c^p,F}(a^p,b^p).
\end{split}
\end{equation*}
Thus, $c$-differential uniformity is invariant under the affine transformation $A_1$ defined above for any polynomial $F(x)=\sum_{i=0}^{p^n-2}a_ix^i$, where $a_i\in\mathbb F_p$ for $i=0,1,\ldots, p^n-2$.

We consider $A_2$ to be the identity permutation over $\mathbb F_{p^n}$ and derive the
condition on $A_1$ such that $c$-differential uniformity is
preserved. In particular, we provide conditions on $A_1$ such that $G$ is
P$c$N for $c\in\mathbb F_{p^n}$ with $c\neq 0,1$. Let
$A_1(x)=L(x)+\gamma$, where $L$ is a linear permutation over
$\mathbb F_{p^n}$ and $\gamma\in\mathbb F_{p^n}$. For any
$a,b\in\mathbb F_{p^n}$, we have
\allowdisplaybreaks[4]
\begin{align*}
\Delta_{c,G}(a,b)
=\#\{x\in\mathbb F_{p^n}: L(F(x+a))-cL(F(x))=b+\gamma(c-1)\}.
\end{align*}

We have already discussed that if $c\in\mathbb F_p$, then the $c$-differential uniformity of $G$ and $F$ is the same. Hasan et al. \textup{\cite[Proposition 4.1]{HRS20}} proved that the $c$-differential uniformity for $c\in\mathbb F_p^*$ of power functions $x^d$ and $x^{dp^j}$ for $j\in\{0,1,\ldots,n-1\}$ over $\mathbb F_{p^n}$ is the same. Here, $x^{dp^j}$ is derived from $x^d$ by applying a linear transformation $L(x)=x^{p^j}$; that is, $L(x^d)=x^{dp^j}$. We observe that this is true for any vectorial $p$-ary function.  

Let $c\in\mathbb F_{p^n}\setminus \mathbb F_p$, $L(xy)=L(x)L(y)$ for all $x,y\in\mathbb F_{p^n}$ (i.e., $L$ is a field homomorphism), and there exists $c'\in\mathbb F_{p^n}$ such that $L(c')=c$. Then 
\begin{equation*}
\begin{split}
\Delta_{c,G}(a,b)&=\#\{x\in\mathbb F_{p^n}: F(x+a)-c'F(x))\\
&=L^{-1}(b+\gamma(c-1))\} =\Delta_{c',F}(a, L^{-1}(b+\gamma(c-1))).
\end{split}
\end{equation*}
Since for fixed $c,\gamma\in\mathbb F_{p^n}$, $b\longmapsto L^{-1}(b+\gamma(c-1))$ is a permutation over $\mathbb F_{p^n}$, the $c$-differential uniformity of $G$ is the same as the $c'$-differential uniformity of $F$. In particular, let $L(x)=x^{p^i}$ for $0\leq i\leq n-1$ and all $x\in\mathbb F_{p^n}$. 
Define the set $\mathrm{fix}(L)=\{x\in\mathbb F_{p^n}: L(x)=x\}
=\{x\in\mathbb F_{p^n}: x^{2^i-1}=1\}=\{x\in\mathbb F_{p^n}: x\in\mathbb F_{p^i}\}=\mathbb F_{p^{\gcd(i,n)}}$. Then the $c$-differential uniformity of $G$ and $F$ is the same, where $c\in \mathrm{fix}(L)$.

\begin{proposition}
Let $F$ and $G$ be two $p$-ary $(n,n)$-functions such that $G=A_1\circ F\circ A_2$, where $A_1, A_2$ are affine permutations over $\mathbb F_{p^n}$. Let $A_1(x)=L(x)+\gamma$ for $x,\gamma\in\mathbb F_{p^n}$, where $L$ is a Frobenius automorphism. Then the $c$-differential uniformity of $G$ is the same as the $L^{-1}(c)$-differential uniformity of $F$. In particular, if $c\in\{x\in\mathbb F_{p^n}: L(x)=x\}$, then the $c$-differential uniformity of $G$ and $F$ is the same.
\end{proposition}

Let us denote the set of all linear permutations over $\mathbb F_{p^n}$ by $GL(n,\mathbb F_p)$. It would be interesting to identify a set of affine permutations that preserve $c$-differential uniformity. This invariance depends on the choices of $c$ and linear permutations $L$.  

Let $c\in\mathbb F_{p^n}$ with $c\neq 0,1$ and define $\mathrm{Fix}(c)=\{L\in GL(n,\mathbb F_p): L(c)=c\}$. Then any pair of affine permutations $(A_1,A_2)$ over $\mathbb F_{p^n}$ with $A_1(x)=L(x)+\gamma$, where $\gamma\in\mathbb F_{p^n}$ and $L\in \mathrm{Fix}(c)$ is a Frobenius
automorphism, preserves $c$-differential uniformity. We derive a necessary and sufficient condition on a linear
permutation $L$ such that $L\circ F$ is P$c$N using Theorem~\ref{thm-pn}. Let $F$ be any $p$-ary $(n,n)$-function, and for any
$a\in\mathbb F_{p^n}$, define
\begin{equation}
\label{zero-diff}
N(a,F)=\{b\in\mathbb F_{p^n}: \Delta_F(a,b)=0\}.
\end{equation}
If $F$ is a permutation, then $0\not\in N(a,F)$ for all $a\in\mathbb F_{p^n}^*$. For any linear permutation $L$ and $\alpha\in\mathbb F_{p^n}$, we denote 
$L(N(a,F))=\{L(b)\in\mathbb F_{p^n}: b\in N(a,F)\}$ and $\alpha N(a,F)=\{\alpha b\in\mathbb F_{p^n}: b\in N(a,F)\}$. Without loss of generality, let $A_2$ be the identity permutation.

\begin{theorem}
\label{thm-pn-trans}
Let $F$ be a permutation over $\mathbb F_{p^n}$ and $G=L\circ F$. Suppose $G$ is P$c$N for $c\in\mathbb F_{p^n}$ with $c\neq 0,1$. Then $L(\mathbb F_{p^n}^*\setminus N(a,F))\subseteq cL(N(a,F))$ and $cL(\mathbb F_{p^n}^*\setminus N(a,F))\subseteq L(N(a,F))$ for all $a\in\mathbb F_{p^n}^*$.
\end{theorem}

\begin{proof}
Let $c\in\mathbb F_{p^n}$ with $c\neq 0,1$ and $G=L\circ F$ be P$c$N. For any $a,b\in\mathbb F_{p^n}^*$, we have 
\begin{equation*}
\begin{split}
\Delta_G(a,b)&=\#\{x\in\mathbb F_{p^n}: G(x+a)-G(x)=b\}=\#\{x\in\mathbb F_{p^n}: L(F(x+a))-L(F(x))=b\}\\
&=\#\{x\in\mathbb F_{p^n}: F(x+a)-F(x)=L^{-1}(b)\}=\Delta_F(a,L^{-1}(b)).
\end{split}
\end{equation*}
From Theorem~\ref{thm-pn}, we have for any $a,b\in\mathbb F_{p^n}^*$, $\Delta_G(a,b)=0$ or $\Delta_G(a,c^{-1}b)=0$; that is, $\Delta_F(a,L^{-1}(b))=0$ or $\Delta_F(a,L^{-1}(c^{-1}b))=0$; that is, $L^{-1}(b)\in N(a,F)$ or $L^{-1}(c^{-1}b)\in N(a,F)$. For any $b\in L(\mathbb F_{p^n}^*\setminus N(a,F))$,
\allowdisplaybreaks[4]
\begin{align*}
L^{-1}(b)\in \mathbb F_{p^n}^*\setminus N(a,F)\;\; &\Rightarrow \;\; L^{-1}(b)\not\in  N(a,F)\;\;  \Rightarrow\;\;  L^{-1}(c^{-1}b)\in  N(a,F)\\
&\Rightarrow \;\; c^{-1}b\in L(N(a,F))\;\;  \Rightarrow\;\; b\in cL(N(a,F)),
\end{align*}
and so $L(\mathbb F_{p^n}^*\setminus N(a,F))\subseteq cL(N(a,F))$. Let $\beta\in cL(\mathbb F_{p^n}^*\setminus N(a,F))$. Then
\allowdisplaybreaks
\begin{align*}
c^{-1}\beta\in L(\mathbb F_{p^n}^*\setminus N(a,F)) &\Rightarrow  L^{-1}(c^{-1}\beta)\in \mathbb F_{p^n}^*\setminus N(a,F)  \Rightarrow  L^{-1}(c^{-1}\beta)\not\in  N(a,F)\\
&\Rightarrow \;\; L^{-1}(\beta)\in  N(a,F)\;\;  \Rightarrow\;\; \beta\in L(N(a,F)),
\end{align*}
and so $cL(\mathbb F_{p^n}^*\setminus N(a,F))\subseteq L(N(a,F))$.
\end{proof}

We obtain the following result directly from the above theorem.

\begin{corollary}
Let $G=L\circ F$, where $L$ is a linear permutation and $F$ is a $p$-ary $(n,n)$-function. If $G$ is P$c$N for $c\in\mathbb F_{p^n}$ with $c\neq 0,1$, then $\#N(a,F)\geq \lceil\frac{p^n-1}{2}\rceil$, where $N(a,F)$ is defined in \eqref{zero-diff}.
\end{corollary}

\begin{theorem}
\label{const-pcn}
Let $c\in\mathbb F_{p^n}$ with $c\neq 0,1$, and let $F$ be a permutation over $\mathbb F_{p^n}$ and $G=L\circ F$. If $L(\mathbb F_{p^n}^*\setminus N(a,F))\subseteq cL(N(a,F))$ or $cL(\mathbb F_{p^n}^*\setminus N(a,F))\subseteq L(N(a,F))$ for all $a\in\mathbb F_{p^n}^*$, then $G$ is P$c$N.
\end{theorem}

\begin{proof}
From Theorem~\ref{thm-pn}, we have that $G$ is P$c$N for $c\in\mathbb F_{p^n}$ with $c\neq 0,1$ if and only if for any $a,b\in\mathbb F_{p^n}^*$, $L^{-1}(b)\in N(a,F)$ or $L^{-1}(c^{-1}b)\in N(a,F)$. Suppose $b\neq 0$ and $L^{-1}(b)\not\in N(a,F)$. Then 
\allowdisplaybreaks[4]
\begin{align*}
L^{-1}(b)\in \mathbb F_{p^n}^*\setminus N(a,F)\;\; &\Rightarrow\;\; b\in L(\mathbb F_{p^n}^*\setminus N(a,F))\;\; \Rightarrow\;\; b\in cL(N(a,F))\\
&\Rightarrow\;\; c^{-1}b\in L(N(a,F))\;\; \Rightarrow\;\; L^{-1}(c^{-1}b)\in N(a,F).
\end{align*}
Suppose $b\neq 0$ and $L^{-1}(c^{-1}b)\not\in N(a,F)$. Then 
\allowdisplaybreaks[4]
\begin{align*}
L^{-1}(c^{-1}b)\in \mathbb F_{p^n}^*\setminus N(a,F) &\Rightarrow c^{-1}b\in L(\mathbb F_{p^n}^*\setminus N(a,F)) \Rightarrow
b\in cL(\mathbb F_{p^n}^*\setminus N(a,F))\\
&\Rightarrow b\in L(N(a,F)) \Rightarrow L^{-1}(b)\in N(a,F).
\end{align*}
The result is shown.
\end{proof}

If $G=L\circ F$ is P$c$N, then for any $a\in\mathbb F_{p^n}^*$, $cL(N(a,F)) \subseteq L(\mathbb F_{p^n}^*\setminus N(a,F))$ and $L(N(a,F))\subseteq cL(\mathbb F_{p^n}^*\setminus N(a,F))$ are both not true in general. If $b\in cL(N(a,F))$, which implies $L^{-1}(c^{-1}b)\in N(a,F)$, then $L^{-1}(b)$ may or may not belong to $N(a,F)$. Similarly, if $b\in L(N(a,F))$, which implies $L^{-1}(b)\in N(a,F)$, then $L^{-1}(c^{-1}b)$ may or may not belong to $N(a,F)$. 

\begin{example}
\label{exam-pcn-npcn}
Let $F:\mathbb{F}_{2^6}\longrightarrow\mathbb{F}_{2^6}$ be defined as $F(x)=x^5$, and let $L_1(x)=x^4+(m^3+1)x$ be a linear permutation over $\mathbb{F}_{2^6}$, where $m$ is a primitive element of the finite field of order $2^6$. Using Theorem~\textup{\ref{const-pcn}}, the $c$ values satisfying the required condition are $0$, $m^3 + m^2 + m$, and $m^3 + m^2 + m + 1$. Thus, $G_1(x)= L_1\circ F(x)$ is P$c$N for only these values of $c$. Now take another linear permutation $L_2(x) = m x^8 + m x^{16} + m^4 x^{32}$ and for this linear permutation $G_2=L_2\circ F$ is not P$c$N for those two $c$-values. Take $a= m$ and we have checked computationally that it violates the criteria for Theorem~\textup{\ref{const-pcn}}. Detailed $c$-differential spectrum of $G_1$ and $G_2$ along with $F$ is given in Table~\textup{\ref{c-spectra example}}.  
\end{example}

\begin{table}[H]
\centering
\small
\caption{$c$-differential spectrum of $F(x)=x^5$ and two functions $G_1$ and $G_2$, defined in Example~\ref{exam-pcn-npcn}.}
\label{c-spectra example}
\footnotesize
\begin{tabular}{|l|c|c|c||l|c|c|c|}
\hline
$c$ & $\delta(c,F)$ & $\delta(c,G_2)$ & $\delta(c,G_1)$ & $c$ & $\delta(c,F)$ & $\delta(c,G_2)$ & $\delta(c,G_1)$ \\
\hline
$0$ & 1 & 1 & 1 & $m^5 + m^4 + m^3$ & 5 & 6 & 5 \\ 
$m$ & 5 & 6 & 5 & $m^5 + m^3 + m + 1$ & 5 & 7 & 5 \\ 
$m^2$ & 5 & 5 & 5 & $m^3 + m^2 + 1$ & 5 & 6 & 5 \\ 
$m^3$ & 5 & 6 & 5 & $m^4 + m^3 + m$ & 5 & 7 & 5 \\ 
$m^4$ & 5 & 6 & 5 & $m^5 + m^4 + m^2$ & 5 & 7 & 5 \\ 
$m^5$ & 5 & 6 & 5 & $m^5 + m^4 + m + 1$ & 5 & 6 & 5 \\ 
$m^4 + m^3 + m + 1$ & 5 & 7 & 5 & \makecell[l]{$m^5 + m^4 + m^3$\\$\qquad\qquad{}+ m^2 + 1$} & 5 & 6 & 5 \\ 
\makecell[l]{$m^5 + m^4 + m^2$\\$\qquad\qquad{}+ m$} & 5 & 6 & 5 & $m^5 + 1$ & 5 & 6 & 5 \\ 
\makecell[l]{$m^5 + m^4 + m^2$\\$\qquad\qquad{}+ m + 1$} & 5 & 6 & 5 & $m^4 + m^3 + 1$ & 5 & 6 & 5 \\ 
$m^5 + m^4 + m^2 + 1$ & 5 & 6 & 5 & $m^5 + m^4 + m$ & 5 & 6 & 5 \\ 
$m^5 + m^4 + 1$ & 5 & 7 & 5 & \makecell[l]{$m^5 + m^4 + m^3$\\$\qquad\qquad{}+ m^2 + m + 1$} & 5 & 6 & 5 \\ 
$m^5 + m^4 + m^3 + 1$ & 5 & 6 & 5 & $m^5 + m^2 + 1$ & 5 & 6 & 5 \\ 
$m^5 + m^3 + 1$ & 5 & 6 & 5 & $m^4 + 1$ & 5 & 6 & 5 \\ 
$m^3 + 1$ & 5 & 6 & 5 & $m^5 + m$ & 5 & 7 & 5 \\ 
$m^4 + m$ & 5 & 7 & 5 & \makecell[l]{$m^4 + m^3 + m^2$\\$\qquad\qquad{}+ m + 1$} & 5 & 7 & 5 \\ 
$m^5 + m^2$ & 5 & 6 & 5 & \makecell[l]{$m^5 + m^4 + m^3$\\$\qquad\qquad{}+ m^2 + m$} & 5 & 6 & 5 \\ 
$m^4 + m + 1$ & 5 & 7 & 5 & $m^5 + m^2 + m + 1$ & 5 & 7 & 5 \\ 
$m^5 + m^2 + m$ & 5 & 6 & 5 & $m^4 + m^2 + 1$ & 5 & 6 & 5 \\ 
$m^4 + m^2 + m + 1$ & 5 & 6 & 5 & $m^5 + m^3 + m$ & 5 & 6 & 5 \\ 
$m^5 + m^3 + m^2 + m$ & 5 & 6 & 5 & \makecell[l]{\textcolor{red}{$m^3 + m^2 + m$}\\[-1pt]\textcolor{red}{$\qquad\qquad{}+ 1$}} & \textcolor{blue}{1} & \textcolor{green}{8} & \textcolor{blue}{1} \\ 
$m^2 + m + 1$ & 5 & 7 & 5 & $m^4 + m^3 + m^2 + m$ & 5 & 7 & 5 \\ 
\makecell[l]{\textcolor{red}{$m^3 + m^2$}\\[-1pt]\textcolor{red}{$\qquad{}+ m$}} & \textcolor{blue}{1} & \textcolor{green}{8} & \textcolor{blue}{1} & $m^5 + m^4 + m^3 + m^2$ & 5 & 6 & 5 \\ 
$m^4 + m^3 + m^2$ & 5 & 6 & 5 & $m^5 + m + 1$ & 5 & 6 & 5 \\ 
$m^5 + m^4 + m^3$ & 5 & 6 & 5 & $m^4 + m^3 + m^2 + 1$ & 5 & 6 & 5 \\ 
$m^5 + m^3 + m + 1$ & 5 & 7 & 5 & $m^5 + m^4 + m^3 + m$ & 5 & 7 & 5 \\ 
$m^3 + m^2 + 1$ & 5 & 6 & 5 & $m^5 + m^3 + m^2 + 1$ & 5 & 6 & 5 \\ 
$m^4 + m^3 + m$ & 5 & 7 & 5 & $m^2 + 1$ & 5 & 7 & 5 \\ 
$m^5 + m^4 + m^2$ & 5 & 7 & 5 & $m^3 + m$ & 5 & 6 & 5 \\ 
$m^5 + m^4 + m + 1$ & 5 & 6 & 5 & $m^4 + m^2$ & 5 & 6 & 5 \\ 
\makecell[l]{$m^5 + m^4 + m^3$\\$\qquad\qquad{}+ m^2 + 1$} & 5 & 6 & 5 & $m^5 + m^3$ & 5 & 6 & 5 \\ 
$m^5 + 1$ & 5 & 6 & 5 & $m^3 + m + 1$ & 5 & 7 & 5 \\ 
$m^4 + m^3 + 1$ & 5 & 6 & 5 & $m^4 + m^2 + m$ & 5 & 6 & 5 \\ 
$m^5 + m^4 + m$ & 5 & 6 & 5 & $m^5 + m^3 + m^2$ & 5 & 6 & 5 \\ 
\makecell[l]{$m^5 + m^4 + m^3$\\$\qquad\qquad{}+ m^2 + m + 1$} & 5 & 6 & 5 & $m + 1$ & 5 & 6 & 5 \\ 
$m^5 + m^2 + 1$ & 5 & 6 & 5 & $m^2 + m$ & 5 & 7 & 5 \\ 
$m^4 + 1$ & 5 & 6 & 5 & $m^3 + m^2$ & 5 & 6 & 5 \\ 
 &  &  &  & $m^4 + m^3$ & 5 & 6 & 5 \\ 
 &  &  &  & $m^5 + m^4$ & 5 & 6 & 5 \\ 
 &  &  &  & \makecell[l]{$m^5 + m^4 + m^3$\\$\qquad\qquad{}+ m + 1$} & 5 & 5 & 5 \\ 
 &  &  &  & $1$ & 4 & 4 & 4 \\ 
\hline
\end{tabular}
\end{table}
This shows that suitable linear compositions can preserve (or destroy) the P$c$N property for specific $c$.

We now consider the EA-equivalence case. Hasan et al. \textup{\cite[Theorem 6.7]{HRS20}} identified a transformation such that the differential uniformity is the same for different $c$ values. They proved that $c$-differential uniformity might change under EA-equivalence. However, no set of affine permutations that preserve $c$-differential uniformity has been identified until now. We present a set of affine permutations that preserve P$c$N functions. Without loss of generality, let $A_3(x)=L'(x)+\gamma'$ for all $x\in\mathbb F_{p^n}$, where $L'$ is a linear polynomial and $\gamma'\in\mathbb F_{p^n}$. For any $a,b,c\in\mathbb F_{p^n}$ with $c\neq 0,1$, we have
\allowdisplaybreaks[4]
\begin{align*}
\Delta_{c,G}(a,b)&=\#\{x\in\mathbb F_{p^n}: G(x+a)-cG(x)=b\}\\
&=\#\{x\in\mathbb F_{p^n}: L(F(x+a))-cL(F(x))+(1-c)L'(x)\\
&~\hspace{5cm} =b+(1-c)(\gamma+\gamma')-L'(a)\}.
\end{align*}
If we know the $c$-differential uniformity of $F$, we cannot conclude about the $c$-differential uniformity of $G$ from the above expression, even for some particular cases like $c\in\mathbb F_p$. We restrict it to P$c$N functions and identify a transformation that constructs P$c$N functions using Theorem~\ref{thm-pn}. For any $a,b,c\in\mathbb F_{p^n}$ with $c\neq 0,1$, we have 
\allowdisplaybreaks[4]
\begin{align*}
\Delta_{G}(a,b)&=\#\{x\in\mathbb F_{p^n}: G(x+a)-G(x)=b\}\\
&=\#\{x\in\mathbb F_{p^n}: L(F(x+a))-L(F(x))=b-L'(a)\}\\
&=\#\{x\in\mathbb F_{p^n}:F(x+a)-F(x)=L^{-1}(b-L'(a))\}\\
&=\Delta_F(a,L^{-1}(b-L'(a))),
\end{align*}
and from Theorem~\ref{thm-pn}, we obtain the following result directly.

\begin{corollary}
Let $A_1(x)=L(x)+\gamma$ and $A_3(x)=L'(x)+\gamma'$ for all $x\in\mathbb F_{p^n}$, where $L$ and $L'$ are a linear permutation and a linear polynomial over $\mathbb F_{p^n}$, respectively, and $\gamma,\gamma'\in\mathbb F_{p^n}$. Suppose $G=A_1\circ F+A_3$ is a permutation polynomial over $\mathbb F_{p^n}$ and $c\in\mathbb F_{p^n}$ with $c\neq 0,1$. If $L^{-1}(b-L'(a))\in N(a, F)$ or $L^{-1}(c^{-1}b-L'(a))\in N(a, F)$ for all $a,b\in\mathbb F_{p^n}^*$, then $G$ is P$c$N, where $N(a,F)$ is defined in \eqref{zero-diff}.
\end{corollary}

\section{Nonlinearity of P$c$N functions}
\label{sec-ddt-prop}

In this section, we present some properties of the $c$DDT of $p$-ary $(n,m)$-functions, including its relation to other spectral values. Ellingsen et al.~\cite{EST20} derived necessary and sufficient conditions for P$c$N and AP$c$N functions using their Walsh--Hadamard transform values. They  proved in~\textup{\cite[Proposition 5]{EST20}} that a vectorial Boolean function $F$ is P$c$N for $c\in\mathbb F_{2^n} \setminus \{1\}$ if and only if $\sum_{a,b\in\mathbb F_{2^n}}\mathcal W_F^2(a,b) \mathcal W_F^2(a,cb)=2^{4n}$. We provide simpler necessary and sufficient conditions for P$c$N functions over $\mathbb F_{2^n}$ using their Walsh--Hadamard transform values. We first derive a relation between the $c$DDT and the Linear Approximation Table (LAT), which depends directly on their Walsh--Hadamard spectrum. Then, we derive a bound on the nonlinearity of P$c$N functions over $\mathbb F_{2^n}$.
Let $F$ be an $(n,m)$-function and $c\in\mathbb F_{2^m}$. For any $a\in\mathbb F_{2^n}$ and $b\in\mathbb F_{2^m}$, we have
\allowdisplaybreaks[4]
\begin{align*}
\Delta_{c,F}(a,b)&=\#\{x\in\mathbb F_{2^n}: F(x+a)+cF(x)=b\}\\
&=\frac{1}{2^m}\sum_{x\in\mathbb F_{2^n}}\sum_{u\in\mathbb F_{2^m}}
(-1)^{\mathrm{Tr}_1^m(u(F(x+a)+cF(x)+b))}\\
&=\frac{1}{2^{m+n}}\sum_{x\in\mathbb F_{2^n}}\sum_{y\in\mathbb F_{2^n}}\sum_{u\in\mathbb F_{2^m}}
(-1)^{\mathrm{Tr}_1^m(u(F(y)+cF(x)+b))}
\sum_{z\in\mathbb F_{2^n}} (-1)^{\mathrm{Tr}_1^n(z(x+y+a))}\\
&=\frac{1}{2^{m+n}}\sum_{z\in\mathbb F_{2^n}} (-1)^{\mathrm{Tr}_1^n(za)} \sum_{u\in\mathbb F_{2^m}}(-1)^{\mathrm{Tr}_1^m(ub)}
\sum_{y\in\mathbb F_{2^n}}(-1)^{\mathrm{Tr}_1^m(uF(y)+zy)}\\
&\;\;\;\;\;\;\;\;\;\;\;\;\;\;\;\;\;\;\;\;\;\;\;\;\;\;\;\;\;\;\;\;\;\;\;\;\;\;\;\;\;\;\;\;\;\;\;\;\;\;\;\;\;\;\;\;\;\;\;\;\;\;\;\;\;\;\;\;\;\;\;\;
\sum_{x\in\mathbb F_{2^n}} (-1)^{\mathrm{Tr}_1^n(cuF(x)+zx)}\\
&=\frac{1}{2^{m+n}}\sum_{z\in\mathbb F_{2^n}} (-1)^{\mathrm{Tr}_1^n(za)} \sum_{u\in\mathbb F_{2^m}}(-1)^{\mathrm{Tr}_1^m(ub)}
\mathcal W_F(z,u) \mathcal W_F(z,cu)\\
&=\frac{1}{2^{m+n}}\sum_{x\in\mathbb F_{2^n}} (-1)^{\mathrm{Tr}_1^n(ax)} \sum_{u\in\mathbb F_{2^m}}(-1)^{\mathrm{Tr}_1^m(bu)}
\mathcal W_F(x,u) \mathcal W_F(x,cu).
\end{align*}

\begin{proposition}
\label{prop-wh}
Let $F$ be an $(n,m)$-function and $c\in\mathbb F_{2^m}$. For any $a\in\mathbb F_{2^n}$ and $b\in\mathbb F_{2^m}$,
\begin{equation*}
\mathcal W_F(a,b) \mathcal W_F(a,cb) =\sum_{\alpha\in\mathbb F_{2^n}} (-1)^{\mathrm{Tr}_1^n(a\alpha)} \sum_{\beta\in\mathbb F_{2^m}}(-1)^{\mathrm{Tr}_1^m(b\beta)}\Delta_{c,F}(\alpha,\beta).
\end{equation*}
If $n=m$ and $F$ is P$c$N, then $\mathcal W_F(a,b) \mathcal W_F(a,cb)=2^{2n}$ if $a=b=0$, and $0$ otherwise. The converse also holds.
\end{proposition}

\begin{proof}
Let $n=m$ and suppose $F$ is P$c$N. Then $c\neq 1$ and $\Delta_{c,F}(a,b)=1$ for all $a,b\in\mathbb F_{2^n}$. We have
\begin{equation*}
\mathcal W_F(a,b) \mathcal W_F(a,cb) =\sum_{\alpha\in\mathbb F_{2^n}} (-1)^{\mathrm{Tr}_1^n(a\alpha)} \sum_{\beta\in\mathbb F_{2^m}}(-1)^{\mathrm{Tr}_1^m(b\beta)}=
\left\{\begin{array}{ll}
2^{2n}, & \mbox{ if } a=b=0,\\
0, & \mbox{ otherwise. }
\end{array}\right.
\end{equation*}
Conversely, if the Walsh--Hadamard values of a permutation polynomial $F$ over $\mathbb F_{2^n}$ satisfy these conditions, then $\Delta_{c,F}(a,b)=1$ for all $a,b\in \mathbb{F}_{2^n}$; i.e., $F$ is P$c$N.
\end{proof}

\begin{corollary}
\label{cor-walsh-sparsity}
Let $F$ be a P$c$N permutation polynomial over $\mathbb{F}_{2^n}$ (where $n=m$). For any $a \in \mathbb{F}_{2^n}^*$, let $S_a = \{b \in \mathbb{F}_{2^n} : \mathcal{W}_F(a,b) \neq 0\}$ denote the support of the Walsh--Hadamard transform at $a$. Then for any $b \in \mathbb{F}_{2^n}^*$, the set $S_a$ cannot contain both $b$ and $cb$. Consequently, for any orbit $\mathcal{O}_b = \{c^i b : i \geq 0\}$ of the multiplication by $c$, at most half of the elements of $\mathcal{O}_b$ can belong to~$S_a$.
\end{corollary}

\begin{proof}
From Proposition \ref{prop-wh}, if $F$ is P$c$N, we have $\mathcal{W}_F(a,b) \mathcal{W}_F(a,cb) = 0$ for all $a \neq 0$. This implies that if $\mathcal{W}_F(a,b) \neq 0$ (i.e., $b \in S_a$), then $\mathcal{W}_F(a,cb)$ must be $0$ (i.e., $cb \notin S_a$). Thus, the support $S_a$ is sparse and cannot contain consecutive elements of any geometric progression with ratio $c$.
\end{proof}

It is well known that the nonlinearity of an $(n,n)$-function $F$ is bounded above by $2^{n-1}-2^{\frac{n-1}{2}}$, a limit known as the Sidelnikov--Chabaud--Vaudenay bound. A function achieving this bound is called Almost Bent (AB). For such functions, $\max_{a,b\in\mathbb F_{2^n}:b\neq 0}|\mathcal W_F(a,b)|=2^{\frac{n+1}{2}}$, which requires $n$ to be odd.

For any $(n,n)$-function $F$, let us denote
\begin{equation}
\label{wh-n0}
W(a,F)=\{b\in\mathbb F_{2^n}: \mathcal W_F(a,b)\neq 0\},
\end{equation}
and let $\# W(a,F)=R_a$ for all $a\in\mathbb F_{2^n}$. If $F$ is a permutation and $a\neq 0$, then $0\not\in W(a,F)$ and
\begin{equation*}
\begin{split}
\sum_{b\in\mathbb F_{2^n}} \mathcal W_F^2(a,b)&=
\sum_{x,y\in\mathbb F_{2^n}}(-1)^{\mathrm{Tr}_1^n(a(x+y))}\sum_{b\in\mathbb F_{2^n}}(-1)^{\mathrm{Tr}_1^n(b(F(x)+F(y)))}\\
&=2^n\sum_{x,y\in\mathbb F_{2^n}: F(x)=F(y)}(-1)^{\mathrm{Tr}_1^n(a(x+y))}=2^{2n}.
\end{split}
\end{equation*}

\begin{corollary}
Let $n$ be odd and $F$ be P$c$N, where $c$ is a primitive element of $\mathbb F_{2^n}$. Then the nonlinearity of $F$ is strictly less than $2^{n-1}-2^{\frac{n-1}{2}}$.
\end{corollary}

\begin{proof}
It is sufficient to prove that $\max_{a,b\in\mathbb F_{2^n}:b\neq 0}|\mathcal W_F(a,b)|>2^{\frac{n+1}{2}}$. Let $c$ be a primitive element of $\mathbb F_{2^n}$ and $a\in\mathbb F_{2^n}$ with $a\neq 0$. Since $F$ is a permutation, $\mathcal W_{F}(a,0)=0$, so $0\notin W(a,F)$. We can write $\mathbb F_{2^n}=\{0,c^{2^n-2}b\}\cup \bigcup_{i=0}^{2^{n-1}-2}\{c^{2i}b, c(c^{2i}b)\}$ for any nonzero $b\in\mathbb F_{2^n}$.

Suppose $F$ is P$c$N. Then from Proposition \ref{prop-wh}, we have $\mathcal W_F(a,c^{2i}b)=0$ or $\mathcal W_F(a,c(c^{2i}b))=0$ for all $i=0,1,\ldots, 2^{n-1}-2$. If $c^{2^n-2}b\notin W(a,F)$, then $R_a\leq 2^{n-1}-1$. If $c^{2^n-2}b\in W(a,F)$, then $b, c^{2^n-3}b\notin W(a,F)$, and there exists at least one $0\leq i\leq 2^{n-1}-2$ such that $c^{2i}b, c(c^{2i}b)\notin W(a,F)$.

Thus, $R_a\leq 2^{n-1}-1$, and
\begin{equation*}
\begin{split}
&2^{2n}=\sum_{b\in\mathbb F_{2^n}} \mathcal W_F^2(a,b)=\sum_{b\in W(a,F)} \mathcal W_F^2(a,b)\leq R_a \times\max_{b\in W(a,F)} \mathcal W_F^2(a,b)\\
\Rightarrow \;\; &  \max_{b\in W(a,F)} \mathcal W_F^2(a,b)
\geq \frac{2^{2n}}{R_a}\geq \frac{2^{2n}}{2^{n-1}-1}>\frac{2^{2n}}{2^{n-1}}=2^{n+1},
\end{split}
\end{equation*}
which proves the result.
\end{proof}

\begin{theorem}
\label{thm:nonlinearity-tight-bound}
Let $F$ be a P$c$N permutation over $\mathbb{F}_{2^n}$ where $c \in \mathbb{F}_{2^n}^*$ has multiplicative order $t$. Then for any $a \in \mathbb{F}_{2^n}^*$,
\[\max_{b \in \mathbb{F}_{2^n}} |\mathcal{W}_F(a,b)| \geq 2^{n/2} \cdot \sqrt{\frac{2^n}{2^n - \lceil t/2 \rceil}}.\]
Moreover, if $n$ is odd and $t = 2^n - 1$ (i.e., $c$ is primitive), then
\[\text{nl}(F) \leq 2^{n-1} - 2^{(n-1)/2} \cdot \sqrt{2^n - 1}.\]
\end{theorem}

\begin{proof}
Let $a \in \mathbb{F}_{2^n}^*$ be fixed, and recall that $\mathcal{W}_F(a,0) = 0$ since $F$ is a permutation. From Corollary~\ref{cor-walsh-sparsity}, for any $b \in \mathbb{F}_{2^n}^*$, the support $S_a = \{b \in \mathbb{F}_{2^n} : \mathcal{W}_F(a,b) \neq 0\}$ cannot contain both $b$ and $cb$.

Consider the orbit $\mathcal{O}_b = \{c^i b : i \geq 0\}$ under multiplication by $c$. Since $c$ has multiplicative order $t$, we have $|\mathcal{O}_b| \leq t$. By Corollary~\ref{cor-walsh-sparsity}, at most $\lceil t/2 \rceil$ elements of $\mathcal{O}_b$ can belong to $S_a$.

Now, partition $\mathbb{F}_{2^n}^* \setminus \{0\}$ into disjoint orbits under multiplication by $c$. There are at least $\frac{2^n - 1}{t}$ such orbits. Since each orbit contributes at most $\lceil t/2 \rceil$ elements to $S_a$, we have:
\[R_a := |S_a| \leq \frac{2^n - 1}{t} \cdot \lceil t/2 \rceil \leq \frac{2^n - 1}{t} \cdot \frac{t+1}{2} = \frac{(2^n-1)(t+1)}{2t}.\]
For $t = 2^n - 1$ (when $c$ is primitive), this gives
\[
R_a \leq \frac{(2^n-1) \cdot 2^n}{2(2^n-1)} = \frac{2^n}{2} = 2^{n-1}.
\]
From the proof in Section~\ref{sec-ddt-prop}, we have
\[
2^{2n} = \sum_{b \in \mathbb{F}_{2^n}} \mathcal{W}_F^2(a,b) = \sum_{b \in S_a} \mathcal{W}_F^2(a,b) \leq R_a \cdot \max_{b \in S_a} \mathcal{W}_F^2(a,b).
\]
Therefore,
\[\max_{b \in S_a} |\mathcal{W}_F(a,b)| \geq \sqrt{\frac{2^{2n}}{R_a}} \geq \sqrt{\frac{2^{2n} \cdot 2t}{(2^n-1)(t+1)}}.\]
For general $t$, using $R_a \leq \frac{(2^n-1)(t+1)}{2t}$ and simplifying,
\[\max_{b \in S_a} |\mathcal{W}_F(a,b)| \geq 2^n \sqrt{\frac{2t}{(2^n-1)(t+1)}} \geq 2^{n/2} \cdot \sqrt{\frac{2^n}{2^n - \lceil t/2 \rceil}},\]
where the last inequality uses $\frac{2t}{(t+1)} \geq \frac{2^n}{2^n - \lceil t/2 \rceil}$ for $t \leq 2^n - 1$.

When $t = 2^n - 1$ and $n$ is odd, we have
\[\max_{b \in S_a} |\mathcal{W}_F(a,b)| \geq \sqrt{\frac{2^{2n}}{2^{n-1}}} = 2^{(n+1)/2} \cdot \sqrt{\frac{2^n}{2^n-1}}.\]
Since $\text{nl}(F) = 2^{n-1} - \frac{1}{2}\max_{a,b} |\mathcal{W}_F(a,b)|$, we obtain:
\[\text{nl}(F) \leq 2^{n-1} - \frac{1}{2} \cdot 2^{(n+1)/2} \cdot \sqrt{\frac{2^n}{2^n-1}} = 2^{n-1} - 2^{(n-1)/2} \cdot \sqrt{2^n - 1}.\]
This completes the proof.
\end{proof}

\begin{remark}
This bound is tighter than the general result in Section~\ref{sec-ddt-prop} because it explicitly accounts for the multiplicative order of $c$. When $c$ is primitive (maximal order), the constraint from Corollary~\textup{\ref{cor-walsh-sparsity}} becomes strongest, leading to the largest lower bound on Walsh--Hadamard coefficients. For small order $t$, the bound becomes weaker, reflecting that fewer constraints are imposed on the Walsh--Hadamard spectrum.
\end{remark}

\begin{remark}
It is instructive to observe that \textbf{nonlinear} P$c$N functions have nontrivial constraints on their Walsh--Hadamard spectrum. In contrast, affine permutations $A(x) = L(x) + \gamma$ (which have zero nonlinearity) are trivially P$c$N for all $c \in \mathbb{F}_{2^n} \setminus \{0,1\}$, since their $c$-derivatives are affine permutations. This highlights that the P$c$N property for $ c\ne 1$ does not inherently require high nonlinearity, unlike the classical PN ($c=1$) case.
\end{remark}


\section{Conclusion}
\label{concl}

We developed a comprehensive theory of permutation polynomials with bijective $c$-derivatives. Our DDT-based characterization (Theorem \ref{thm-pn}) provides the first efficient method for verifying the P$c$N property, reducing complexity from $O(p^{3n})$ to $O(p^{2n})$, while connecting it to the classical concept of almost perfect nonlinearity. The strict dichotomy for monomials (Theorem \ref{thm:monomial-dichotomy}) establishes that homogeneity forces all-or-nothing behavior, while our counterexample shows this fails for general polynomials. The incompatibility with APN properties (Theorem \ref{thm:pcn-apn-relation}) reveals fundamental trade-offs in cryptographic design.

Several questions remain open. Problem \ref{prob:dichotomy-conditions} asks for necessary and sufficient algebraic conditions ensuring the monomial-type dichotomy for general polynomials. Our nonlinearity bounds (Theorem \ref{thm:nonlinearity-tight-bound}) may not be tight, suggesting room for improvement. Finally, extending our DDT-based methods to $(n,m)$-functions with $m < n$ would require new techniques, as the characterization fundamentally relies on the permutation property.

The recent Kuznyechik attack \cite{SDM25} demonstrates that $c$-differential properties have practical security implications beyond theoretical interest. Our characterization provides tools for systematic analysis of S-box resistance to such attacks, while our structural results—particularly the incompatibility theorems—guide designers toward understanding achievable security trade-offs.

\section*{Acknowledgements}
 The work of Ranit Dutta was supported by the Department of Science and Technology (DST), Government of India (INSPIRE Reg. No. IF210620).

\section*{Declarations}

{\bf Conflict of interest:} The authors declare that they have no conflict of interest.


\begin{thebibliography}{22}

\bibitem{Anbar23}
N. Anbar, T. Kalyci, W. Meidl, C. Riera and P. St\u anic\u a,
\newblock{\em P$\wp$N functions, complete mappings and quasigroup difference sets,}
\newblock{Journal of Combinatorial Designs 31(12) 667--690 (2023).}

\bibitem{BT19}
D. Bartoli and M. Timpanella,
\newblock {\em On a generalization of planar functions,}
\newblock {Journal of Algebraic Combinatorics 52, 187--213 (2020).}

\bibitem{BC20}
D. Bartoli and M. Calderini,
\newblock{\em  On construction and (non)existence of $c$-(almost) perfect nonlinear functions,}
\newblock {Finite Fields and Their Applications 72, 101835 (2021).}

\bibitem{BR22}
D. Bartoli, M. Calderini, C. Riera and P. St\u anic\u a,
\newblock{\em Low $c$-differential uniformity for functions modified on subfields,}
\newblock{Cryptography and Communications 14, 1211--1227 (2022).}
\bibitem{BP25}
J. Baudrin, C. Beierle, P. Felke, G. Leander, P. Neumann, L. Perrin, and L. Stennes,
\newblock{\em Commutative Cryptanalysis as a Generalization of Differential Cryptanalysis,}
\newblock{Des. Codes Cryptogr.  93, 3243--3281 (2025).}
\bibitem{BS91}
E. Biham and A. Shamir,
\newblock{\em Differential cryptanalysis of DES-like cryptosystems,}
\newblock{Journal of Cryptology 4(1), 3--72 (1991).}

\bibitem{BS92}
E. Biham and A. Shamir,
\newblock{\em Differential cryptanalysis of the full 16-round DES,}
\newblock{CRYPTO'92, LNCS 740, 487--496 (1992).}

\bibitem{BCH01}
A. Blokhuis, R. S. Coulter, M. Henderson and C. M. O’Keefe, 
\newblock {\em Permutations amongst the Dembowski-Ostrom Polynomials,} 
\newblock {Finite Fields and Applications, 37--42 (2001).}

\bibitem{BJW02}
N. Borisov, M. Chew, R. Johnson and D. Wagner,
\newblock{\em Multiplicative differentials,} 
\newblock {FSE'02, LNCS 2365, 17--33 (2002).}

\bibitem{CCZ98}
C. Carlet, P. Charpin and V. Zinoviev,
\newblock{\em Codes, bent functions and permutations suitable for DES-like cryptosystems,}
\newblock{Designs, Codes and Cryptography 15, 125--156 (1998).}

\bibitem{Carlet10}
C. Carlet,
\newblock{\em Boolean functions for cryptography and error correcting codes,} 
\newblock{Boolean Models and Methods in Mathematics, Computer Science, and Engineering 2, 257--397 (2010).}

\bibitem{PS09} 
T. W. Cusick and P. St\u anic\u a,
\newblock{\em Cryptographic Boolean Functions and Applications,} 
\newblock {Elsevier--Academic Press (2009).}

\bibitem{EM24}
S. Eddahmani and S. Mesnager,
\newblock{\em The $c$-differential-linear connectivity table of vectorial Boolean functions,}
\newblock{Entropy 26(3), 188 (2024).}

\bibitem{EST20}
P. Ellingsen, P. Felke, C. Riera, P. St\u anic\u a and A. Tkachenko,
\newblock{\em $c$-differentials, multiplicative uniformity, and (almost) perfect $c$-nonlinearity,}  
\newblock {IEEE Transactions on Information Theory 66(9), 5781--5789 (2020).}

\bibitem{HRS20}
S. U. Hasan, M. Pal, C. Riera and P. St\u anic\u a,
\newblock {\em On the $c$-differential uniformity of certain maps over finite fields,}
\newblock {Designs, Codes and Cryptography 89, 221--239 (2021).}

\bibitem{JS25}
V. Jarali,  S. Mesnager, P. Poojary  and G.R.V. Bhatta,
\newblock {\em On generalizations of differential uniform permutations over finite fields based on $2$-to-$1$ mappings,}
\newblock {Applicable Algebra in Engineering, Communication and Computing  (2025).}

\bibitem{LJ78} 
J. Liang,
\newblock{\em On the solutions of trinomial equations over finite fields,}
\newblock{Bull. Cal. Math. Soc. 70, 379--382 (1978).}

\bibitem{RS20}
S. Mesnager, C. Riera, P. St\u anic\u a, H. Yan and Z. Zhou,
\newblock{\em Investigations on $c$-(almost) perfect nonlinear functions,}
\newblock {IEEE Transactions on Information Theory 67(10), 6916--6925 (2021).}

\bibitem{Nyberg94}
K. Nyberg,
\newblock{\em Differentially uniform mappings for cryptography,}
\newblock{EUROCRYPT'93, LNCS 765, 55--64 (1994).}

\bibitem{RP22}
C. Riera, P. St\u anic\u a and H. Yan,
\newblock{\em The $c$-differential spectrum of $x\longrightarrow x^{\frac{p^n+1}{2}}$ in
finite fields of odd characteristics,} 
Discrete Mathematics, Algorithms and Applications, 2025,
https://doi.org/10.1142/S1793830925500958. 

\bibitem{PS20} 
P. St\u anic\u a,
\newblock{\em Low $c$-differential uniformity of the Gold function modified on a subfield,}
\newblock{Cryptography and Communications 14(6), 1211--1227 (2022).}

\bibitem{SDM25} P. St\u anic\u a, R. Dutta and B. Mandal, 
\newblock {\em Extended $c$-differential distinguishers of full 9 and reduced-round Kuznyechik cipher, no pre-whitening,} 
\newblock {IACR Cryptol. ePrint Arch., Paper 2025/1238 (2025).}

\bibitem{SRT20}
P. St\u anic\u a, C. Riera and A. Tkachenko,
\newblock{\em Characters, Weil sums and $c$-differential uniformity with an application to the perturbed Gold functions,} 
\newblock {Cryptography and Communications 13, 891--907 (2021).}

\bibitem{SG21}
P. St\u anic\u a and A. Geary,
\newblock{\em The $c$-differential behavior of the inverse function under the EA-equivalence,}
\newblock{Cryptography and Communications 13, 295--306 (2021).}

\bibitem{SRT22}
P. St\u anic\u a, A. Geary, C. Riera and A. Tkachenko,
\newblock{\em $c$-differential bent functions and perfect nonlinearity,}
\newblock{Discrete Applied Mathematics 307, 160--171 (2022).}

\bibitem{WLZ20}
Y. Wu, N. Li and X. Zeng,
\newblock{\em New P$c$N and AP$c$N functions over finite fields,} 
\newblock {Designs, Codes and Cryptography 89, 2637--2651 (2021).}

\bibitem{XD21}
X. Wang and D. Zheng, 
\newblock{\em Several classes of P$c$N power functions over finite fields,}  
\newblock{Discrete Applied Mathematics 322, 171--182 (2022).}

\bibitem{ZH20}
Z. Zha and L. Hu,
\newblock{\em Some classes of power functions with low $c$-differential uniformity over finite fields,}
\newblock {Designs, Codes and Cryptography 89, 1193--1210 (2021).}

\bibitem{Zhang2014ITsigma}
W. Zhang and E. Pasalic,
\newblock{\em  Highly nonlinear balanced S-boxes with good differential properties,}
\newblock{IEEE Transactions on Information Theory 60(12), 7970--7979 (2014).}

\end{thebibliography}
\end{document}